\documentclass[11pt,showkeys]{revtex4-2}
\usepackage[utf8]{inputenc}

\usepackage{mathtools} 
\usepackage{amsmath}
\usepackage{color}
\usepackage{amsfonts}
\usepackage{natbib}
\usepackage{amsthm}
\usepackage{url}
\usepackage{hyperref}
\usepackage{graphicx}
\usepackage{quotes}  
\usepackage{boxedminipage}
\usepackage{nccmath}
\usepackage{bbm}
\usepackage{times}
\usepackage[T1]{fontenc}
\usepackage{float}
\usepackage{enumitem}
\renewcommand{\theenumi}{(C\arabic{enumi})}
\usepackage{booktabs}
\usepackage{thmtools}
\usepackage{subcaption}
\usepackage{multirow}
\usepackage{url}
\usepackage{pgfplots}
\pgfplotsset{compat=newest}
\usepackage{bm}
\usepackage{footnote}
\newcommand{\bfitDelta}{\bm{\mathit{\Delta}}}

\usepackage[ruled, lined, linesnumbered, commentsnumbered, longend]{algorithm2e}

\SetCommentSty{mycommfont}

\usepackage{makecell, cellspace, caption}
\usepackage{subcaption}
\usepackage{array}
\newcolumntype{L}[1]{>{\raggedright\let\newline\\\arraybackslash\hspace{0pt}}m{#1}}
\newcolumntype{C}[1]{>{\centering\let\newline\\\arraybackslash\hspace{0pt}}m{#1}}
\newcolumntype{R}[1]{>{\raggedleft\let\newline\\\arraybackslash\hspace{0pt}}m{#1}}

\newcommand{\comment}[1]{}

\DeclareMathOperator*{\argmax}{arg\,max}
\DeclareMathOperator*{\argmin}{arg\,min}

\theoremstyle{plain}

\newtheorem{lemma}{Lemma}

\theoremstyle{definition}
\newtheorem{theorem}{Theorem}
\newtheorem{definition}{Definition}
\newtheorem{example}{Example}
\newtheorem{remark}{Remark}

\usepackage[margin=1in]{geometry}

\usepackage{tikz}
\usepackage{tikz-3dplot}

\renewcommand\thmcontinues[1]{Continued}

\tdplotsetmaincoords{60}{110}
\pgfmathsetmacro{\rvec}{.8}
\pgfmathsetmacro{\wvec}{30}
\pgfmathsetmacro{\phivec}{60}

\begin{document}

\title{Feature-Based Network Construction: From Sampling to What-if Analysis}

\author{Christian Franssen}
\email{c.p.c.franssen@vu.nl}
 \affiliation{Dept. Operations Analytics, VU Amsterdam, The Netherlands. }

\author{Joost Berkhout}%
 \email{joost.berkhout@vu.nl}
\affiliation{Dept. of Mathematics, VU Amsterdam, The Netherlands.}

\author{Bernd Heidergott}%
 \email{b.f.heidergott@ivu.nl}
\affiliation{Dept. Operations Analytics, VU Amsterdam, The Netherlands.}%

\begin{abstract}
Networks are characterized by structural features, such as degree distribution, triangular closures, and assortativity.
This paper addresses the problem of reconstructing instances of continuously (and non-negatively) weighted networks 
from given feature values.
We introduce the gradient-based Feature-Based Network Construction (FBNC) framework.
FBNC allows for sampling networks that satisfy prespecified features exactly (hard constraint sampling).
Initializing the FBNC gradient descent with a random graph, FBNC can be used as an alternative to exponential random graphs in sampling graphs conditional on given feature values.
We establish an implicit regularization approach to the original feature-fitting loss minimization problem so that FBNC achieves a parsimonious change in the underlying graph, where the term "implicit" stems from using appropriate norms in the very construction of the FBNC gradient descent.
In constructing the implicit regularization, we distinguish between the case where weights of a link can be chosen from a bounded range, and, the more demanding case, where the weight matrix of the graph constitutes a Markov chain.
We show that FBNC expands to 
"what-if analysis" of networks, that is, for a given initial network and a set of features satisfied by this network, FBNC finds the network closest to the initial network with some of the feature values adjusted or new features added.
Numerical experiments in social network management and financial network regulation demonstrate the value of FBNC for graph (re)construction and what-if analysis.
\end{abstract}

\keywords{Network construction, steepest feasible descent, network what-if analysis, implicit regularization}

\maketitle

\section{Introduction}

Networks, broadly defined as systems of interconnected entities or nodes, serve as a powerful framework for modeling complex relationships across a wide range of domains. 
In this paper we identify networks through their weight matrices $ W $, where $ W_{ i j } \geq 0  $ denotes the weight of link $ ( i , j ) $, with $ 1 \leq i , j \leq N $, with $N$ denoting the number of nodes in the network.
In the literature, networks are conveniently classified and analyzed using {\em structural features}, or just \textit{features} in short.
Formally, any feature $ \Phi $ is a mapping of $ W $ to $ \Phi ( W ) \in \mathbb{R}^d $, for $ d \geq 1 $.
Examples of some well-known and widely used features include the \textit{node strengths} \citep{barrat2004architecture}, {\em triangular closures} and {\em modularity} expressing the extent of clustering \citep{luce1949method,Newman_2006,Leicht_Newman_2007}, {\em assortativity} for describing the tendency for nodes to be connected with nodes that are similar \citep{newman2002assortative}, {\em centrality} providing information about the relative importance of nodes and links \citep{bonacich1987power,borgatti2005centrality}, and the {\em Kemeny constant} 
\citep{kemeny1976finite, berkhout2019analysis} or
\textit{effective graph resistance}, resp. \textit{Kirchhoff index} \citep{ghosh2008minimizing,bianchi2019kirchhoffian}, which can be used for measuring network connectivity.

In this paper, we show how to construct non-negative and finite weighted digraphs satisfying prespecified features \textit{exactly}, that is, we introduce an approach to numerically solve the inverse problem  of 
finding $ W $ so that $ \Phi_i ( W ) = \phi_i $, for $i=1,\dots, m$, for some array of $m$ features $ (\Phi_1 , \ldots, \Phi_m ) $ and feature values $(\phi_1 , \ldots, \phi_m )$.
We introduce a gradient descent approach, called  Feature Based Network Construction (FBNC), for numerically computing this inverse.

{\bf Hard Constraint Sampling.} 
FBNC enables hard constraint sampling of networks,  that is, FBNC samples networks that satisfy prespecified features exactly.
Hard constraint sampling offers several key applications. 
First, it allows for \textit{generating network instances with limited data}; see, e.g., \citet{Mastrandrea_Squartini_Fagiolo_Garlaschelli_2014, Squartini_Caldarelli_Cimini_Gabrielli_Garlaschelli_2018, Engel_Pagano_Scherer_2019} for related works.
When detailed information is unavailable due to privacy concerns (common in social and financial networks), FBNC allows the creation of networks that meet high-level, known features. Secondly, it is applicable for \textit{imputing missing network data}, which is important since network data is often incomplete in practice. FBNC generates plausible data that adheres to known high-level features for a given network type. Thirdly, it can be applied for \textit{random network generation} \citep{coolen2017generating} by simply randomly initializing the gradient descent FBNC algorithm. 

Hard constraint sampling differs fundamentally from the weak constraint sampling approach, such as exponential random graphs \citep{holland1981exponential,snijders2002markov,handcock2010modeling}, which is predominant in the literature. The term  "weak" 
means that the sampled networks satisfy the desired feature only on average, resulting in a (macro)-canonical ensemble. In contrast, hard constraint sampling requires that each sampled network satisfies the constraint exactly. This strict adherence characterizes hard constraint sampling as drawing from the microcanonical ensemble. For a comprehensive discussion and relevant literature, see Section~\ref{sec:application1}.
Throughout this article, we refer to "hard constraint sampling" simply as "sampling" unless additional clarification is needed.

{\bf Implicitly Regularized Hard Constraint Sampling.}
A key observation for our approach is that the equivalence class of networks with identical feature values is typically large. 
Hence, the image of the inverse mapping is typically a large set.
This opens the opportunity for encoding preference relations for network realizations when computing the inverse.
A powerful instance of such a preference relation is to find networks in the solution set that are "close" to some given network. 
Here, next to finding $ W $ that such that $ \Phi_i ( W ) = \phi_i $, for $i=1,\dots, m$, for some array of features $\Phi_i $ and corresponding values $\phi_i$, we wish to simultaneously minimize the $ L^1 $ or $ L^2 $-distance of $ W $  to a given network $ W_0 $.
In the literature, this is commonly dealt with by {\em explicit} regularization where a regularizing term is added to the objective function, such as Ridge or Lasso regularization \citep{hoerl1970ridge, lasso}.
However, we show that this approach does not produce networks satisfying the hard feature constraints.  
To overcome this problem, we incorporate the closeness preference by ensuring the gradient FBNC's descent trajectory is {\em implicitly} regularized via efficient steepest feasible descent computations in both the $L^1$ and $L^2$-norm. 
As we show, using the steepest feasible descent in the $L^1$-norm results in descent steps exclusively changing the weights of links that decrease the loss most effectively.
This is particularly interesting as the network can establish a newly desired feature while only changing a limited set of links.
We call this an implicitly regularized hard constraint sampling approach as the hard constraint sampling is still guaranteed, while directing the algorithm to prioritize minimal network changes whenever feasible. 

Being able to implicitly regularize network changes, our FBNC algorithm allows for another key application: {\em feature-based what-if analysis} of networks. 
By generating feature-satisfying networks that do not change too much from an existing network, our FBNC algorithm allows for meaningful what-if analyses of networks for various feature values.  For example, we may study how a given network can be changed to increase its connectivity efficiently (for which we provide examples in Section~\ref{sec:what-if}).
Moreover, the resulting network trajectory of our FBNC algorithm can be a road map to realize an additional network feature. We consider this to be a paradigm shift in (gradient) descent-driven optimization, where not only the final solution of an algorithm is made fruitful, but its entire trajectory provides valuable insights. 

The paper is organized as follows.
Section~\ref{sec:technical_fbnf} presents the technical analysis of FBNC and the algorithmic theory is developed in Section~\ref{sec:algorithm}.
We continue by demonstrating how to use our FBNC algorithm for random network sampling
in Section~\ref{sec:application1}. 
In Section~\ref{sec:application2}, we present
a numerical example showing how the what-if analysis can help financial regulators assess the role of exposure diversification on systemic risk in financial networks.
We briefly conclude in Section~\ref{sec:discussion}.

\section{Feature-Based Network Construction} \label{sec:technical_fbnf}

In this section, we formally introduce the FBNC framework.
We start by briefly introducing the required notation in Section~\ref{sec:not}.
Then, we outline some of the most commonly used features in Section~\ref{sec:features} and we provide the FBNC problem formalizations in Section~\ref{sec:probformulation}.
We conclude with a discussion on (implicit) regularization approaches for solving the FBNC problem in Section~\ref{sec:ir}.

\subsection{Preliminary notation} \label{sec:not}

Denote an unweighted directed network by $G=(V,E)$, where $V = \{1,\dots,N\}$ denotes the set of $N$ \textit{nodes}, $E \subseteq V \times V$ the set of \textit{links} between nodes.
For $i \in V$, we denote by $ V_i = \{ j \in V  : (i,j )  \in E \}  $ the set of out-neighbors of $i$.
In the following, we assume that $ V_i $ contains at least one node $ j $ with $ i \not = j $; in words, $ i $ has at least one node different from itself as an out-neighbor.
Similarly, we denote the set of in-neighbors of $i$ by $ {_i } V = \{  j \in V :  (j,i )  \in E \}  $.
We extend this framework  by introducing weights $W_{ i , j } \geq 0 $ for links $ ( i , j ) \in E $,
where we let $ W_{ i ,j } = 0 $ for $ ( i , j ) \not \in E $.
Note that for $ ( i , j ) \in E $ the weight may be zero (i.e., $ W(  i , j )  = 0 $), and we interpret this as the link $ ( i , j ) $ being present but not ``active'', which is in contrast to $ ( i , j )  \not \in E$ as this means that link $ ( i , j )$ is not present and cannot be made active by increasing the weight.
Recall that with slight abuse of notation, we often use $W$ and graph $G = (V,E,W)$ as synonyms and omit notation representing possible dependency on $(V,E)$ for simplicity.

\subsection{A Review of (Structural) Features} \label{sec:features}

For a given network $G = (V,E)$ (or $W$), a feature function can be defined on the unweighted topology of the network, that is, on the link set $E$, or alternatively, on the weights of this topology, which are confined to $W$.
Note that while an infinite space of feature functions essentially exists, network theory typically concentrates on a limited set of predominantly utilized network features.
Two of the most commonly used feature functions are the in-degree $\delta_i^- = | _i V |$ and out-degree $\delta^+_i = |V_i | $ counting the number of links going towards and leaving a node $i$, respectively (or in the case of undirected networks, simply the degree).
We define the weighted versions of the in-strengths 
\begin{equation}\label{eq:s-}
 s_i^- = \sum_{j \in _i V} W_{ j , i} 
\end{equation}
as the sum of link weights pointing towards $i$, 
and the out-strengths as the sum of link weights leaving a node given by
\begin{equation}\label{eq:s+}
 s_i^+ = \sum_{j \in  V_i} W_{ i,j} 
\end{equation}
see \citep{barrat2004architecture}.

In certain contexts, it is desirable to model a graph using a weighted adjacency matrix $W$ that satisfies the conditions of stochasticity, whereby its entries are non-negative, and the sum of each row is equal to 1. 
Transition matrix $W$ then models the behavior of a random walker on the graph. 
This is a commonly employed method; for example, see \citet{patel2015robotic} and \citet{berkhout2019analysis}. From $W$, we can derive specific features yielding insights into the network structure, such as the stationary distribution and the Kemeny constant.

Table~\ref{tab:feature_examples} provides a selection of commonly used features, as well as a short description and a reference to their mathematical definitions in Appendix~\ref{sec:sf}, which are suitable for the network construction algorithm presented in Section~\ref{sec:algorithm}.

\begin{table}
\centering
\caption{Overview of commonly used features functions. Features with a $^\dag$ are defined using the transition matrix satisfying the conditions of stochasticity.}
\begin{tabular}{llll} 
            \hline
            \textbf{Feature function} \hspace{5pt} & \textbf{Symbol} \hspace{5pt} & \textbf{Description} & \textbf{Definition} \\
            \hline
            In-strength &$s^- = (s_i^-)_{i \in V}$ \hspace{5pt}& The sum of link weights reaching each node $i$. &  (\ref{eq:s-}) \\
            \hline
            Out-strength &$s^+ = (s_i^+)_{i \in V}$  & The sum of link weights leaving each node $i$. &   (\ref{eq:s+}) \\
            \hline
            Reciprocity &$r$  & The extent to which nodes reciprocate incoming \hspace{5pt} & Appendix \ref{sf:reciprocity}\\
             & & links.& \\
            \hline
            Triangular & $c$  & The extent to which nodes form closed triplets. & Appendix \ref{sf:3cycles} \\
            closures & & & \\
            \hline
            Modularity &$Q$ & The extent to which nodes are grouped in  & Appendix \ref{sf:mod} \\
            & & communities or clusters. &  \\
            \hline
            Assortativity &$\rho$ & Tendency for nodes to link with others that are  & Appendix \ref{sf:assortativity} \\
            & & similar in node covariates or a function over the   \\
            & & network itself (e.g., node strengths). & \\
            \hline
            Stationary &$\pi = (\pi_i)_{i \in V}$ & The long-run fraction of time the random walker  & Appendix \ref{sf:stat}\\
             $\text{distribution}^\dag$& & spends at each node $i$. &   \\ 
            \hline
            Kemeny  &$K$ & The expected time to reach a $\pi$-randomly drawn  & Appendix \ref{sf:kemeny1} \\
            $\text{constant}^\dag$ & & node from any arbitrary node. & \\
            \hline
            Effective graph \hspace{3pt} &$R$ & The sum of electrical resistances in a resistor  & Appendix \ref{sf:effective_resistance} \\ 
            resistance & & network with edge weights as conductances. &  \\
            \hline
\end{tabular}
\label{tab:feature_examples}
\end{table}

\subsection{The Hard Constraint Sampling Problem}\label{sec:probformulation}

In the context of optimization (with $ G = (V,E) $ fixed), let $w$ be the control variable that is a (vector) parametrization of the weights in $W$.
For $ i \in V$, we denote by 
\[
 w_{(i, \cdot )  }    = ( W_{ i, j } : j \in V_i )
\] 
the vector of weights in the $ i$th row of $ W $ that are feasible for optimization.
The concatenation of all such row vectors  $ w_{ ( i , \cdot ) }$ is denoted by 
\begin{eqnarray} \label{eq:parameterization}
    w =  w_{ (1 , \cdot)} \mathbin\Vert  w_{(2, \cdot )} \mathbin\Vert \: \dots  \:\mathbin\Vert w_{ ( |V| , \cdot ) }  ,
\end{eqnarray}
where $ x \mathbin\Vert  y $ denotes the concatenation to vectors $ x $ and $ y $.
This allows to express a weight matrix $ W$ through a corresponding weight vector $ w \in \mathbb{R}^{|E|}$. To that end, let $e_{i} \in \mathbb{R}^N$ denote the $i$-th basis vector in $\mathbb{R}^N$ and let $(\cdot)'$ denote the transposition of a row-vector. The weight matrix can be retrieved from $ w $ in a unique way by 
\[
W = \sum_{(i,j) \in E } w_{i,j} e_{i} e_{j}',
\]
where for simplicity, we write $ w_{ i , j }$ for the element in (vector) $ w$ that corresponds to entry $ ( i , j ) $ of $ W$, for $ ( i , j ) \in E$.

For our analysis, we distinguish between the \textit{bounded setting},
where a hypercube gives the set of feasible values of $ W $,
and the more restrictive \textit{Markov setting}, where 
$ W $ is only feasible if $ W $ is a stochastic matrix.
This is introduced in the following definition.

\begin{definition} \label{def:W}
The set of feasible weight matrices is denoted by $ {\cal W} \in \{ \mathcal{W}^{b}, \cal{W}^M\}$.
\begin{itemize}

\item In the {\em bounded setting}, $ \cal W$ is given as hypercube $ {\cal W}^b := [0 , b ]^{|E|} $, for $ b > 0 $.

\item In the {\em Markov setting}, $ {\cal W} $ is the set ${\cal W}^M$ of all stochastic matrices over $ ( V, E ) $, i.e.,
\[
{\cal W}^M : = \{ w \in [ 0 , 1 ]^{|E|} : \sum_{j \in V_i} w_{i , j } =1, \text{ for all } i \in V \},
\]
and we call graph $W \in \mathcal{W}^M$ a \textit{Markovian graph}.
\end{itemize}
\end{definition}

The goal of FBNC is to enforce several features in a network, such that the constructed network $\Hat{W}$ satisfies $\Phi_i ( \hat W ) = \phi^*_i$, $ i = 1, \ldots, m$, for the feature functions $\Phi_1 , \ldots , \Phi_m$ and feature values $\phi^*_1 , \ldots , \phi^*_m$.
This can be expressed as minimization of the loss function
\begin{align}\label{eq:lossf}
    J ( W , \Phi  , \phi^*) := \sum_{i=1}^m \left\| \Phi_i ( W ) - \phi_i^* \right\|_2^2,
\end{align}
where for ease of presentation, we assume that the target values in $\phi^*$ are scalar.
Note that we will often interchangeably use $w$ and $W$, thus define feature functions $\Phi(w) = \Phi(W)$, and use $J(w)$ to refer to $J(w, \Phi, \phi^*)$ throughout the rest of the paper.
Thus, in terms of $w$ instead of $W$,  FBNC constructs networks from the set
\begin{align} \label{eq:pf_theta}
{\cal W}_{\Phi,\phi^*} := \argmin_{w \in {\cal W} } \ J(w, \Phi, \phi^*),
\end{align}
where $ {\cal W} \in \{ {\cal W}^b , {\cal W}^M\}$ is the set of feasible weight values,
and $\Phi$ is assumed to be a smooth map.
If features $\Phi$ with values $\phi^*$ are graphical, then by (\ref{eq:pf_theta})
$$
{\cal W}_{\Phi,\phi^*} = \{ w \in {\cal W} : J(w,\Phi,\phi^*) = 0\}  \not = \emptyset .
$$
It follows that a gradient-based descent algorithm is generally sufficient for solving \eqref{eq:pf_theta} (we will discuss our algorithmic approach in Section~\ref{sec:algorithm}). 
In case non-global minima are found, one can reject the solution and restart with different initial weights.

A key observation is that one is typically interested in only a handful of features (we will demonstrate this using the numerical examples in this paper) which renders $ {\cal W}_{\Phi,\phi^*}$ to be a large set.
Hence, when applying a gradient descent algorithm to finding a minimizer for (\ref{eq:pf_theta}), the 
limiting network of the gradient descent solution will depend on the initial value.
Figure~\ref{fig:trajectories} shows this schematically, ${w}^1$ and ${w}^2$ 
represent two distinctive initial network instances, the dotted line present the path of the gradient descent algorithm, and ${w}_1$ and ${w}_2$ the respective limiting values.

\begin{figure}[ht]
\centering
\begin{tikzpicture}[->,>=stealth',auto,node distance=3cm]
    \draw (3.0,1.0) circle (1.4) 
    (3.45,1.35) node [text=black,above]{$\mathcal{W}_{\Phi,\phi^*}$}
    (-1.0,-1.0) rectangle (5,3) 
    (-0.75,2.5) node [text=black,above]{$\mathcal{W}$};
    
    \coordinate[label = above:$w^1$] (A) at (0,0);
    \node at (A)[circle,fill,inner sep=1pt]{};

    \coordinate[label = above:$w^2$] (B) at (1,2);
    \node at (B)[circle,fill,inner sep=1pt]{};

    \coordinate[label = above:${w}_1$] (C) at (3.05,0.2);
    \node at (C)[circle,fill,inner sep=1pt]{};

    %

    \coordinate[label = below:${w}_2$] (D) at (2.4,1.8);
    \node at (D)[circle,fill,inner sep=1pt]{};
    \node at (0.6,-0.4) [circle,inner sep=1pt]{}; 
    \node at (3.5,2.4) [circle,inner sep=1pt]{}; 
    \path[every node/.style={font=\sffamily\small}]
        (A) edge[bend right, dashed] node [below right] {} (C);
    %

    \path[every node/.style={font=\sffamily\small}]
        (B) edge[bend left, dashed] node [below right] {} (D);
        
    \end{tikzpicture}
    \caption{Schematic illustration of two networks obtained through adjusting link-weights so that target feature values are met.
    Both networks ${w}_1$ and ${w}_2$ depend on the initial network given to a gradient-based algorithm.}
    \label{fig:trajectories}
\end{figure}
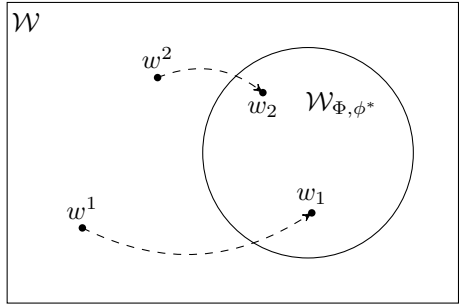

Note that $ w_1 $ and $w_2 $ solve the hard constraint sampling problem.
The fact that the limiting values $ w_i $ depend on the initial values $ w^i$ allows for a straightforward application of this approach to randomly sampling from $ {\cal W}_{\Phi,\phi^*} $.
Indeed, randomly initializing the gradient-based descent algorithm yields random limiting points, and for details we refer to Section~\ref{sec:application1}. 

\subsection{The Regularized Hard Constraint Sampling Problem} \label{sec:ir}

For the specific purpose of what-if analysis, the additional goal is to minimize the changes in the network while meeting a new target feature. 
To formalize, let us view the collection of graphs $ ( V, E, {\cal W} )$ as a normed vector space $\mathcal{V}_p := \left(\mathbb{R}^{|E|}, \| \cdot \|_p\right)$, where $p \in \{1, 2\}$, so that $ \| \cdot \|_1 $ refers to the taxicab norm, and $ \| \cdot \|_2 $ refers to the Euclidean norm.
Furthermore, let $w(0)$ denote the weights of the initial network.
Similar to \eqref{eq:pf_theta}, our primal objective is to generate a graph from the target set ${\cal W}_{\Phi,\phi^*}$.
Secondarily, we strive to minimize the difference between the starting graph $w(0)$ and 
the solution $w$: $\| w - w(0) \|_p$, $p \in \{1,2\}$, as to regularize the network changes.
This admits the following problem formulation
\begin{align} \label{eq:what_if_pf}
    \begin{array}{rl}
    \min \quad &  \| w - w(0)\|_p \\
    \textup{s.t.} \quad & w \in {\cal W}_{\Phi,\phi^*} ,
    \end{array}
\end{align}
with $  {\cal W}_{\Phi,\phi^*} $ given via  (\ref{eq:pf_theta}).
This formulation essentially defines the problem as a projection in $ \| \cdot \|_p $ onto the space ${\cal W}_{\Phi,\phi^*}$. 
Due to the complex structure of $ {\cal W}_{\Phi,\phi^*} $, a solution to problem (\ref{eq:what_if_pf}) is generally very difficult to find, except for special cases.
This motives the search for approximate solutions of (\ref{eq:what_if_pf}).

A common heuristic for approximately solving \eqref{eq:what_if_pf} is through adding a penalty term to the objective function:
\begin{align} \label{eq:pf_theta_reg_exp}
\argmin_{w \in {\cal W} } \ J(w) + \beta \left\|w - w(0)\right\|_p^p,
\end{align}
for some $\beta >0$, also known as Ridge regression in case $ p=2 $, see, e.g., \citet{ridge1,ridge2}, and Lasso regression when $ p =1 $, see, e.g., \citet{lasso}.
This approach is referred to as \textit{explicit regularization} in the literature.
However, as we illustrate in Section~\ref{sec:gh}, the solution set of \eqref{eq:pf_theta_reg_exp} does not align with the target set $\mathcal{W}_{\Phi, \phi^*}$, and in fact, any solution of \eqref{eq:pf_theta_reg_exp} will only approach the target set, see Figure~\ref{fig:trajectories_wi} for a schematic illustration. Thus, this approach is infeasible for solving \eqref{eq:what_if_pf}.
An alternative and more fruitful approach for addressing \eqref{eq:what_if_pf} is using a steepest feasible descent algorithm, where the steepest feasible descent direction is found within the desired normed vector space $\mathcal{V}_p$.
As we will detail in Section~\ref{sec:algorithm}, we can do so by defining the steepest feasible descent in the desired norm $L^p$-norm, e.g., see \citet{byrd1975extension}.
Then, using such a steepest descent in the FBNC algorithm, the distance $\|w - w(0)\|_p$ is \textit{implicitly regularized} as at every step $k$ the loss function $J(w(k))$ is greedily minimized and so guides the trajectory along a path that minimizes distance within $\mathcal{V}_p$. Numerical examples illustrating this behaviour are presented in Section~\ref{sec:algorithm}.

\begin{figure}[ht]
\centering
\begin{tikzpicture}[->,>=stealth',auto,node distance=3cm]
    \draw (3.0,1.0) circle (1.4) 
    (3.45,1.35) node [text=black,above]{$\mathcal{W}_{\Phi,\phi^*}$}
    (-1.0,-1.0) rectangle (5,3) 
    (-0.75,2.5) node [text=black,above]{$\mathcal{W}$};
    
    \coordinate[label = below:$w(0)$] (A) at (-0.35,0);
    \node at (A)[circle,fill,inner sep=1pt]{};

    \coordinate[label = {[label distance = 0.005pt] right:$\hat{w}_{\mathrm{imp}}$}] (Z) at (1.6,1);
    \node at (Z)[circle,fill,inner sep=1pt]{};

    \coordinate[label = below:$\hat{w}_{\mathrm{exp}}$] (C) at (1.5,0.45);
    \node at (C)[circle,fill,inner sep=1pt]{};

    \node at (3.5,2.4) [circle,inner sep=1pt]{};

    \path[every node/.style={font=\sffamily\small}]
        (A) edge[bend left, dashed] node [below right] {} (C);
    \path[every node/.style={font=\sffamily\small}]
       (A) edge[bend left, dashed] node [below right] {} (Z);


\end{tikzpicture}

    \caption{Schematic illustration of two networks obtained through adjusting link-weights 
    so that target feature values are met while applying explicit and implicit regularization. Here, $\hat{w}_{\mathrm{explicit}}$ is the result of using an explicit regularization and does not reach the target set. Instead, an implicitly regularized $\hat{w}_{\mathrm{implicit}}$ does reach the target set.}
    \label{fig:trajectories_wi}
\end{figure}
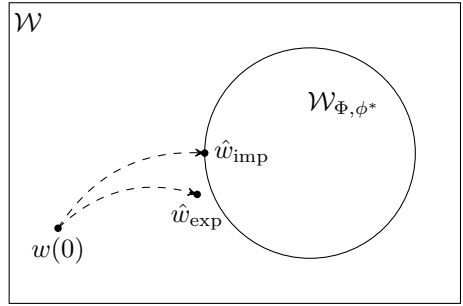

We conclude by remarking that when ${\cal W}_{\Phi,\phi^*}$ can be represented by linear (or linearizable) constraints, we can formulate problem \eqref{eq:what_if_pf} using (Mixed-Integer) Linear Programming. However, many network features exhibit non-linearity, rendering this approach impractical.

\section{Descent Algorithms for FBNC} \label{sec:algorithm}

In this section, we will detail our algorithmic approach to solve the FBNC problems in \eqref{eq:pf_theta} and \eqref{eq:what_if_pf}. 
Our steepest feasible descent algorithm is specifically designed for solving \eqref{eq:what_if_pf}, but since \eqref{eq:what_if_pf} is an extension of \eqref{eq:pf_theta}, we also use this algorithm for solving \eqref{eq:pf_theta}.

This section is organized as follows.
First, we introduce the steepest feasible descent algorithm and establish bounds on step sizes to ensure feasibility. 
In Section~\ref{sec:gh}, we explain how to compute the steepest feasible descent direction for graphs defined on hypercubes, considering both the taxicab and Euclidean norms. 
We also demonstrate that explicit regularization of the loss function is not viable within our framework. 
Additionally, we highlight the importance of small step sizes in our steepest descent algorithm, contrary to the common practice of using larger step sizes in steepest feasible descent methods in the existing literature.
We then proceed with explaining how to compute the steepest feasible descent direction for Markovian graphs, considering both the taxicab and Euclidean norms in Section~\ref{sec:MG}.
We conclude by showing how we can use a projection to approximate the steepest feasible descent direction in Section~\ref{sec:proj_sfd}.

\subsection{Steepest Feasible Descent}

For given $ \Phi $ and $ \phi^* $, the loss function $ J $, defined in (\ref{eq:lossf}), is a mapping from 
$\left( {\cal W}, \| \cdot \|_p\right)$ onto $\left(\mathbb{R}, \| \cdot \|\right) $.
We will impose throughout the paper the following condition:
\begin{itemize}
    \item [{\bf (A1)}]
    {\em $ J $ is (Fréchet) differentiable in all $w \in \{ w \in {\cal W}$ : $J(w) < \infty\}$, where $\mathcal{W} = \mathcal{W}^b$ in the bounded setting and $\mathcal{W} = \mathcal{W}^M$ in the Markov setting.}
\end{itemize}
Fréchet differentiability is ensured by the existence of a bounded linear operator $C(w)$ such that 
\begin{align} \label{eq:lim}
    \lim_{z \in \mathcal{Z}, z \neq w, z \to w} \frac{\| J( z ) - J( w ) - C (w) (z-w)) \|_p}{\|z-w \|_p} = 0, 
\end{align}
for $p=2$, and $\mathcal{Z} = \mathcal{Z}^b := \mathbb{R}^{|E|}$ in the bounded setting or $\mathcal{Z} = \mathcal{Z}^M := \{z \in \mathbb{R}^{|E|} : \sum_{j\in V_i} z_{i,j} = 1, \forall i\in V\}$ in the Markov setting.
Note that if \eqref{eq:lim} holds for $p=2$, the same $C(w)$ for $p=1$ ensures \eqref{eq:lim} is satisfied.
Condition {\bf (A1)} implies $ J ( w ) $ is a well-defined mathematical object not only on $ {\cal W}   $ but also in a neighborhood around $ {\cal W}$'s boundary. 
For simple features such as in- and out-strengths, this is easy to see.
For features based on the Markov setting (i.e., $ {\cal W} = {\cal W}^M$), we note that while $ w \in \mathcal{Z}^M \setminus {\cal W}^M$ has no interpretation as Markov chain (as $z$ contains at least one negative element), the analytical formulas for features and their derivatives extend beyond $ {\cal W}^M$.  
For a discussion of this and for explicit derivative expression using matrix calculus, we refer to \citet{Conlisk_1985} and \citet{caswell2019sensitivity}.

Furthermore, let $\nabla J(w) := C(w)$ in \eqref{eq:lim}. 
We require the following condition, which will be assumed throughout the paper:

\begin{itemize}

\item[{\bf (A2)}] $ \nabla J(w) $ is continuous in all $w \in \{ w \in {\cal W}$ : $J(w) < \infty\}$.

\end{itemize}
Combining {\bf (A1)} and {\bf (A2)}, we obtain that $J$ is continuously differentiable in all $w \in \{ w \in {\cal W}$ : $J(w) < \infty\}$, which we will use for the convergence of the steepest feasible descent algorithm presented later in this section.

Some features in Table~\ref{tab:feature_examples} do not immediately satisfy {\bf (A1)} and {\bf (A2)} . 
For example, modularity and assortativity, see Appendix \ref{sf:mod} and \ref{sf:assortativity}, are not defined for all in $W \in \mathcal{W}$ and therefore are not continuously differentiable everywhere.
In the case of modularity, it sufficient for continuous differentiability to assume that the graph has a positive sum of weights, i.e., we can impose an arbitrary small $\varepsilon > 0$ on all links $(i,j) \in E$. 
In the case of assortativity, imposing non-identical in-strengths ensures that a slightly adapted version of {\bf (A1)} and {\bf (A2)} are satisfied (we will elaborate more on this later in this section).
Moreover, the stationary distribution $\pi$ (Appendix \ref{sf:stat}) exclusively exists if $w$ induces an irreducible Markovian graph and therefore is not differentiable for all $w \in \mathcal{W}^M$.
The use of stationary distribution $\pi$ as feature can be allowed by assuming that $G = (V,E)$ is strongly connected and ensuring a Markovian graph $w$ is irreducible.
To achieve this, we can also impose an arbitrarily small $\varepsilon > 0$ on all links $(i,j) \in E$, so that {\bf (A1)} and {\bf (A2)} are satisfied.
Note that that the Kemeny constant $K$ (Appendix \ref{sf:kemeny1}) does satisfy {\bf (A1)} and {\bf (A2)} as continuous differentiability is only violated when $w \in \mathcal{W}^M$ induces a reducible Markovian graph. However, for these $w$ it holds $J(w) = \infty$ and hence such $w$ are no limiting points of any steepest decent algorithm. A similar argument can be made for the effective graph resistance (Appendix \ref{sf:effective_resistance}).

For $ w \in {\cal W}$, we denote the set of \textit{feasible directions} for $ w \in \mathcal{W}$ as 
\[
 \bfitDelta ( w ; \mathcal{W}, p ) := 
\{  d \in \mathbb{R}^{|E|} :   \: \| d \|_p \leq 1 \land \exists \hat \eta :   w + \eta d \in {\cal W } , \: 0 < \eta \leq \hat \eta \},
\]
where in the following, we will simply write $\bfitDelta ( w ; p )$.
In words, one can transverse from $ w $ along $ d \in \bfitDelta ( w ; p )  $ for a sufficiently small distance without leaving $ \cal W $.
For any feasible direction $ d $ at $ w$, we can define 
the \textit{directional derivative} along $ d $ through
\begin{align} \label{eq:dd}
\partial_d J(w ) = \lim_{\eta \to 0} \frac{J\left( w + \eta d \right) - J(w)}{\eta} .
\end{align}
A {\em steepest feasible descent} at $ w $ for norm $ p $ can now be defined through
\begin{align} \label{eq:sfd}
\delta( w , p )  \in \argmin_{ d \in \bfitDelta ( w ; p )  } \partial_d J(w) .
\end{align}
In case the set on the above right-hand side has more than one element, we choose $ \delta( w ; p )$ randomly from this set.
Now that we have introduced $ \delta ( w , p ) $, we can define FBNC formally.

\begin{definition}{\bf (FNBC and trajectory)}
The FBNC algorithm is a algorithmic mapping $ \cal F$, where $ \cal F$ is built as a steepest descent algorithm. To that end, we choose $ {\cal W} \in \{ {\cal W}^b,  {\cal W}^M\} $ and $p \in \{ 1, 2 \}$, and we denote the maximum number of descent steps by $ T \in \mathbb{N}$, the step size sequence by $\{ \alpha(k) \}$ with $  \alpha(k) > 0 $, and the loss function $ J ( w ) $ as in (\ref{eq:lossf}) for some feature functions $ \Phi$ with target feature values $\phi^*$.
For $ w \in {\cal W} $, the algorithmic mapping is specified by
\[
{\cal F}_{  (\mathcal{W}, p, J) } ( w ) :=  {\cal F}_{ (  \mathcal{W}, p , J ; \alpha , T , \gamma ) } ( w ) =: w (T ),
\]
with $\gamma$ denoting a convergence parameter, which, if not noted otherwise, is used to terminate the algorithm if $\|\Phi_i(w) - \phi_i^*\|_2 < \gamma$, for $i = 1, \ldots, m$.
We call the network sequence $ \{ w ( k ) : k = 0, \ldots, T \}$, 
with $ w ( 0 ) = w$ and
\begin{align}\label{eq:descentalgo}
   w (k+1 ) &= w ( k ) + \alpha(k) \delta ( w ( k ) , p )   \in {\cal W} , \qquad \forall k = 0, 1, 2, \dots, T-1 ,
\end{align}
the \textit{path} or \textit{trajectory} of $ {\cal F}_{ (\mathcal{W}, p, J) } ( w )$ (resp. FBNC), 
where $ \delta ( w ( k ) , p ) $ is a steepest feasible descent direction for $ p $ as in \eqref{eq:sfd} (again, we suppress $\mathcal{W}$ for notational convenience). 
\end{definition}

For the step size sequence $\{ \alpha(k) \}$ we take a fixed (generally small) $\alpha$ to implicitly regularize the trajectory $\{w(k)\}$. We will further elaborate on the need for small step sizes in Section~\ref{sec:gh}.
To ensure the feasibility of $\{ w(k) \}$, we bound the step size by a maximal step size possible.
This is an effective way of enforcing the non-negativity constraints of $\mathcal{W}$ and commonly applied in methods such as the Frank-Wolfe algorithm \citep{frank1956algorithm}.

For bounding the step size, we let $\alpha > 0$ denote a fixed constant (generally small if regularization is required). 
Assume $\delta( w ( k ) , p ) \neq 0$ so that for the bounded setting, i.e., $\mathcal{W} = \mathcal{W}^b$, we use
\begin{align} \label{eq:ss_bounded}
   \alpha(k) =  \min \biggl(\alpha, \min_{(i,j) \in E : \delta( w ( k ) , p )_{i,j}  \neq 0}  \biggl( \frac{w(k)_{i,j}}{-\delta ( w ( k ) , p )_{i,j}}\mathbbm{1}_{\{ \delta ( w ( k ) , p )_{i,j} < 0 \}} + \frac{b - w(k)_{i,j}}{\delta ( w ( k ) , p )_{i,j}}\mathbbm{1}_{\{ \delta ( w ( k ) , p )_{i,j} > 0 \}} \} \biggr) \biggr),
\end{align}
and in the Markov setting, i.e., $\mathcal{W} = \mathcal{W}^M$, we use
\begin{align} \label{eq:ss_markov}
    \alpha(k) = \min\left( \alpha,  \min_{(i,j) \in E: \delta( w ( k ) , p )_{i,j}   \neq 0} \left( \frac{w(k)_{i,j}}{-\delta ( w ( k ) , p )_{i,j}}\mathbbm{1}_{\{ \delta ( w ( k ) , p )_{i,j} < 0 \}} \} \right)\right),
\end{align}
where the second terms in both (outer) $\min\{\cdot, \cdot\}$ operators ensure feasibility of the trajectory $\{w(k)\}$. 

Moreover, to ensure implicit regularization, we let for each $k = 0,1,2,\dots, T-1$,
\begin{align} \label{eq:ds}
J(w (k )) \leq J( w ( k ) + \alpha(k) \delta ( w ( k ) , p ) ).
\end{align}
In words, at every $k$ we want to take a \textit{descent step}.
A standard numerical approach to ensure a descent step is the \textit{Armijo rule}, that for given $\alpha(k) > 0$, $\beta \in (0,1)$, and $\sigma \in (0,1)$, scales the step size according to $\beta^{\kappa} \alpha(k)$ \citep{armijo1966minimization, bertsekas1976goldstein}, where $\kappa$ is the first non-negative integer such that
the decrease $J(w(k)) - J(w(k+1))$, for $w(k+1) = w(k) + \beta^{\kappa} \alpha(k) \delta ( w ( k ) , p )$, is larger than the (scaled) linear approximation in $w(k+1)$, i.e., 
\begin{align} \label{eq:armijo}
J(w(k)) - J(w(k) + \beta^{\kappa} \alpha(k) \delta ( w ( k ) , p )) \geq - \sigma \beta^{\kappa} \alpha(k) \nabla_{\delta ( w ( k ) , p )} J(w(k)).
\end{align}
Here, $\nabla_{\delta ( w ( k ) , p )} J(w(k))$ is the directional derivative in the steepest feasible descent direction $\delta ( w ( k ) , p )$.
The Armijo rule ensures a natural trade-off between the distance between $w(k)$ and $w(k+1)$ in each step $k$ and the decrease of the loss function (hence helping in implicit regularization). 
Convergence is assured by satisfying {\bf (A1) } and {\bf (A2) } \citep{bertsekas1976goldstein}, where parameters $\sigma$ and $\beta$ are the tuning parameters that will be fixed $\sigma = \beta = \frac{1}{2}$ \citep{armijo1966minimization}, if not noted otherwise.

It is known that under conditions {\bf (A1) } and {\bf (A2)}, $ w(T) $ converges towards a stationary point, see Definition~\ref{def:stat_point}, of $ J ( w ) $ as $  T $ tends to $ \infty $, for which we provide an exact statement in the following in Theorem~\ref{thm:convergence}. 

\begin{definition}\label{def:stat_point}
    We call any $w \in \mathcal{W}$, such that $\delta(\Hat{w},p) = 0$, a stationary point of $J(w)$.
\end{definition}

\begin{theorem} \label{thm:convergence}
Let ${\cal W}$ be defined as in Definition \ref{def:W}. 
Furthermore, let $ J : {\cal W} \mapsto \mathbb{R}$ satisfy condition {\bf (A1)} and {\bf (A2)} and $J(w(0)) < \infty$. Then,
\[
\Hat{w} := \lim_{ T \rightarrow \infty } w (T ) .
\]
is a stationary point of $ J $, i.e., $\delta(\Hat{w},p) = 0$.
\end{theorem}
\begin{proof}
For a complete formal proof, we refer to \citet{bertsekas1997nonlinear}. Here, we give a brief sketch of the proof.
Let $\Hat{w}$ be the limiting point of $\{w(k)\}$.
By continuity of $J(w)$, $\{J(\Hat{w})\}$ is the limiting point of $\{J(w(k))\}$.
By the Armijo rule, we have that $\{ J(w(k)) \}$ is strictly monotone decreasing.
Therefore, assuming $J(w(0)) < \infty$, it follows that $J(w)$ is continuously differentiable in $\{ w(k) \}$ as follows from Conditions {\bf (A1)} and {\bf (A2)}.
Moreover,
\[
J(w(k)) - J(w(k+1)) \to 0,
\]
for $k \to \infty$, which implies
\[
- \sigma \beta^{\kappa} \alpha(k) \nabla_{\delta ( w ( k ) , p )} J(w(k)) \to 0,
\]
since $J$ is bounded below by 0.
Note that $\beta^\kappa \not \to 0$ by convergence of the Armijo rule. 
Moreover, $\alpha(k) \to 0$ can only occur if the steepest feasible descent direction $\delta ( w ( k ) , p ) \to 0$ by construction.
Therefore, $\nabla_{\delta ( w ( k ) , p )} J(w(k)) \to 0$, which implies $\delta ( w ( k ) , p ) \to 0$ and so $\{w(k)\}$ converges to a stationary point.
\end{proof}

If $  \delta ( w ( T ) , p ) = 0$ in \eqref{eq:descentalgo}, then $ w ( T ) $  is a stationary point for the loss minimization problem in \eqref{eq:pf_theta}. It may be the case that $w(T)$ is not a global minimum, i.e., $J(w(T)) > 0$, for example, because ${\cal W}_{\Phi , \phi^*} = \emptyset$, i.e., the feature setting $ \Phi = \phi^*$ is not graphical.
Applying $ {\cal F}_{ (\mathcal{W},p,J) }(w) $ to random initializations, we may test for graphicality; see the following remark.

\begin{remark} \label{rmk:non_graphicality}
{\em Let $w_{l}$ be randomly initialized using some probability density function 
$f_{ {\cal W}}(w)$, 
for which $f_{ {\cal W}} (w) > 0$, for all $ w \in {\cal W} $ and $f_{ {\cal W}}(w) = 0$ for all $w \notin {\cal W}$. 
If $$\min \{ J ( {\cal F}_{ ( \mathcal{W}, p , J )}(w_{1})), \dots, J( {\cal F}_{ ( \mathcal{W},p , J )  }(w_{l}))\} > 0,$$ for $l \xrightarrow[]{} \infty$, then
\begin{align*}
    \mathbb{P} \big (   J( {\cal F}_{ ( \mathcal{W},p , J ) }(w)) > 0 \big  ) \xrightarrow[]{} 1, \ \text{for all} \ w \in {\cal W} , 
\end{align*}
and the problem is not graphical.
}
\end{remark}

We will proceed with the computation of the steepest feasible descent direction for graphs in the bounded and Markov settings considering both the taxicab and Euclidean norms. 
Finding the steepest descent direction in optimization involves solving an optimization problem, as discussed in previous works like \cite{zoutendijk1960methods} and \cite{Rosen_1960}, and more recently summarized by \cite{chen2000methods}.
As we will show, it is simple to determine the steepest feasible descent direction in the taxicab norm if the steepest feasible direction in the Euclidean norm is available. 
It is worth noting that \citet{gafni1984two} already present efficient methods for finding the steepest feasible descent direction in the Euclidean norm, especially when dealing with non-negativity or simplex constraints on control variables. 
In both settings, our result follows from a straightforward extension of their work.

\subsection{Steepest Feasible Descent for Graphs over a Hypercube} \label{sec:gh}

In this subsection, we consider the steepest feasible descent for graphs in $ {\cal W} = {\cal W}^b$. 
To that end, we denote by $u_{ i, j }$ a vector of size $|E|$ with value $1$ at the entry corresponding to $ (i, j ) \in E$, and $0$ otherwise.

\begin{theorem}\label{ref:thmdescent}
Let $w \in {\cal W}^b $, then a steepest feasible descent direction is given
\begin{itemize}
\item[(i)] for $ p=2$, by
$\delta(w,2) = \left(\frac{\Tilde{\delta}(w,2)_{i,j}}{\| \Tilde{\delta}(w,2) \|_2}\right)_{(i,j) \in E}$, where 
\[
 \Tilde{\delta}( w , 2 )_{i,j} =  \begin{cases} 
  0 & (i,j) \in \mathcal{A};\\ 
 -\nabla J(w)_{i,j} & (i,j) \notin \mathcal{A},
\end{cases}
 \]
with $$\mathcal{A} = \{ (i,j) \in E : w_{i,j} = 0 \text{ and } \nabla J(w)_{i,j} > 0 \text{, or } w_{i,j} = b \text{ and } \nabla J(w)_{i,j} < 0  \}.$$

\item[(ii)]
for $ p =1 $, by
\[
 \delta ( w , 1 ) =
 \begin{cases}
     u_{ i, j } & \text{if } \delta ( w , 2 )_{i,j} > 0; \\
     - u_{ i, j } & \text{if } \delta ( w , 2 )_{i,j} < 0; \\
     0 & \text{if } \delta ( w , 2 )_{i,j} = 0, 
 \end{cases}
\]
where 
\[
( i , j ) \in \argmax_{i,j}  \ | \delta ( w , 2 ) |,
\]
and $|\cdot|$ transforms the vector to absolute values. 
In the case that $\argmax_{i,j}  \ | \delta ( w , 2 ) |$ is not a singleton set, we simply choose $(i,j) \in \argmax_{i,j}  \ | \delta ( w , 2 ) |$ randomly.
\end{itemize}

\end{theorem}
\begin{proof}
Let $p=2$.
Assuming $\mathcal{A} = \emptyset$, clearly $\delta(w,2)$ points in the opposite direction of $\nabla J(w)$ and thereby is simply the normalized negative gradient.
Now assume that $\mathcal{A}$ is not empty.
Finding the steepest feasible descent direction involves a projection of $-\nabla J(w)$ on the subspace $\{d \in \mathbb{R}^{|E|} : d_{i,j} = 0 \text{ for all } (i,j) \in \mathcal{A}\}$, where $$\mathcal{A} = \{ (i,j) \in E : w_{i,j} = 0 \text{ and } \nabla J(w)_{i,j} > 0 \text{, or } w_{i,j} = b \text{ and } \nabla J(w)_{i,j} < 0  \},$$
which is the set of elements that would violate the hypercube constraints $0 \leq w_{i,j}$ or $w_{i,j} \leq b$ in the direction of $-\nabla J(w)$.
This projection is achieved by setting the direction to $0$ for all elements in $\mathcal{A}$.
By projection, we preserve the optimality of the direction in its directional derivative, see \eqref{eq:dd}.
Normalization of the direction $\tilde{\delta}( w , 2 )$ to unit length, i.e., $\delta( w , 2 ) = \tilde{\delta}( w , 2 ) / \| \tilde{\delta}( w , 2 )\|_2$, also does not affect optimality in the directional derivative, so it follows that $\delta( w , 2 )$ is the steepest feasible descent direction.

We proceed by considering $p=1$, where we leverage the obtained $\delta( w , 2 )$.
Given the constraint $\| \delta( w , 1 ) \|_1 \leq 1$, we know that a corner solution minimizes $\partial_d J(w) = d \cdot \nabla J(w)$ by the linearity of the problem.
Thus, we can simply choose the element $(i,j)$ of $\delta( w , 2 )$ with the steepest potential decrease in the loss function $J(w)$ and let $\delta( w , 1 ) = u_{ i,j} $ if $\delta(w,2)_{i,j} > 0$ and $\delta( w , 1 ) = -u_{ i,j} $ if $\delta(w,2)_{i,j} < 0$.
Finally, when $\delta(w,2)_{i,j} = 0$, we simply let $\delta( w , 1 ) =0$.
Therefore, it follows that $\| \delta( w , 1 ) \|_1 \leq 1$ is satisfied.
Note that the direction in $(i,j)$ is feasible by the construction of $\delta( w , 2 )$.
\end{proof}

\vspace{0.25cm}

In slight abuse of notation, we refer to the steepest feasible descent in (i) 
of Theorem~\ref{ref:thmdescent}
in a subsuming way {\em $ L^2$-descent}, and the steepest descent in (ii) 
of Theorem~\ref{ref:thmdescent}
as {\em $L^1$-descent}.

As Theorem \ref{ref:thmdescent} shows, the choice of $ p$ impacts the steepest descent direction and thus leads to different implicit regularizations.
Moreover, choosing $ p=1$ leads to a steepest descent that greedily chooses the best link weight to adjust, thereby having a much more parsimonious impact on the number of links that are adjusted than the standard steepest descent with $ p =2 $ that potentially changes all links simultaneously.

In the literature, regularization is typically done explicitly by adding a regularizing term to the objective function.
More specifically, the minimization of $ J ( w ) $ is replaced by the minimization of the form in \eqref{eq:pf_theta_reg_exp}.
As Example \ref{ex:reg_explicit} shows, such explicit regularization renders network construction with hard constraints carried out by the FBNC algorithm infeasible.

\begin{example} \label{ex:reg_explicit}
{\em Let us consider the set of networks with node set $V = \{1,2\}$, 
link set $E = \{(1,2), (2,1)\}$, and link weights $w = ( w_{1,2} ,w_{2,1} )  $ with $ w \in {\cal W}^b = [ 0, 1 ]^2 $.
Consider the feature
\[
\Phi(w ) =  w_{1,2} + 2 w_{2,1} ,
\]
and suppose we want to construct graphs such that $ \Phi ( w ) = 1$.
This leads to the loss function
\[
J ( w ) = \left\| w_{1,2} + 2 w_{2,1} - 1 \right\|_2^2 .
\]
For explicit regularization we add the regularizer $ \beta ( \|w - w(0)\|_p \big )^p$ to $ J ( w) $
for some $ \beta > 0$. We choose $\beta = 1$ and consider the extended loss function
\[
\hat J_p ( w ) = J( w ) + \|w - w(0)\|_p^p,
\]
for $ p = 1 ,2 $.

In Figure~\ref{fig:er2}-\ref{fig:ir1}, we visualize the paths of ${\cal F}_{ ( [ 0 , 1]^2,p , J ; \alpha , T , \gamma ) } ( w ) $ for $w \in {\cal W}^b$ for various choices for the loss function and norm, and choose $ T=10000 $, $ \alpha=10^{-3}$, and $\gamma=10^{-3} $.
More specifically, by choosing an initial point $ w $ in $ [ 0 , 1]^2$ and following the green lines (indicating the steepest descent direction), we can trace the path the algorithm will follow.
The red line denotes the target set.

Figure~\ref{fig:er2} shows the trajectories of $ {\cal F}_{ ([ 0 , 1]^2, 2 , \hat J_2  )} ( w ) $, i.e., algorithm (\ref{eq:descentalgo}) applied to 
$ \hat J_2(w) $.
Figure~\ref{fig:er2} shows that the explicit regularization approach prevents the algorithm from reaching the target set.
This is opposed to Figure~\ref{fig:ir2}, where we apply $ {\cal F}_{ ([ 0 , 1]^2, 2 , J ) } (w ) $ which uses implicit regularization via the $ L^2$-norm and finds the target set while minimizing the distance of $\{w(k)\}$ to $w(0)$.

Figure~\ref{fig:er1} and Figure~\ref{fig:ir1} show the same effect when instead of $ \|w - w(0)\|^2_2$ the factor $  \| w - w(0)\|^1_1$ is used for explicit regularization in the $L^1$-norm.
\begin{figure}
\centering
\begin{minipage}{.48\textwidth}
  \centering
  \includegraphics[width=\linewidth]{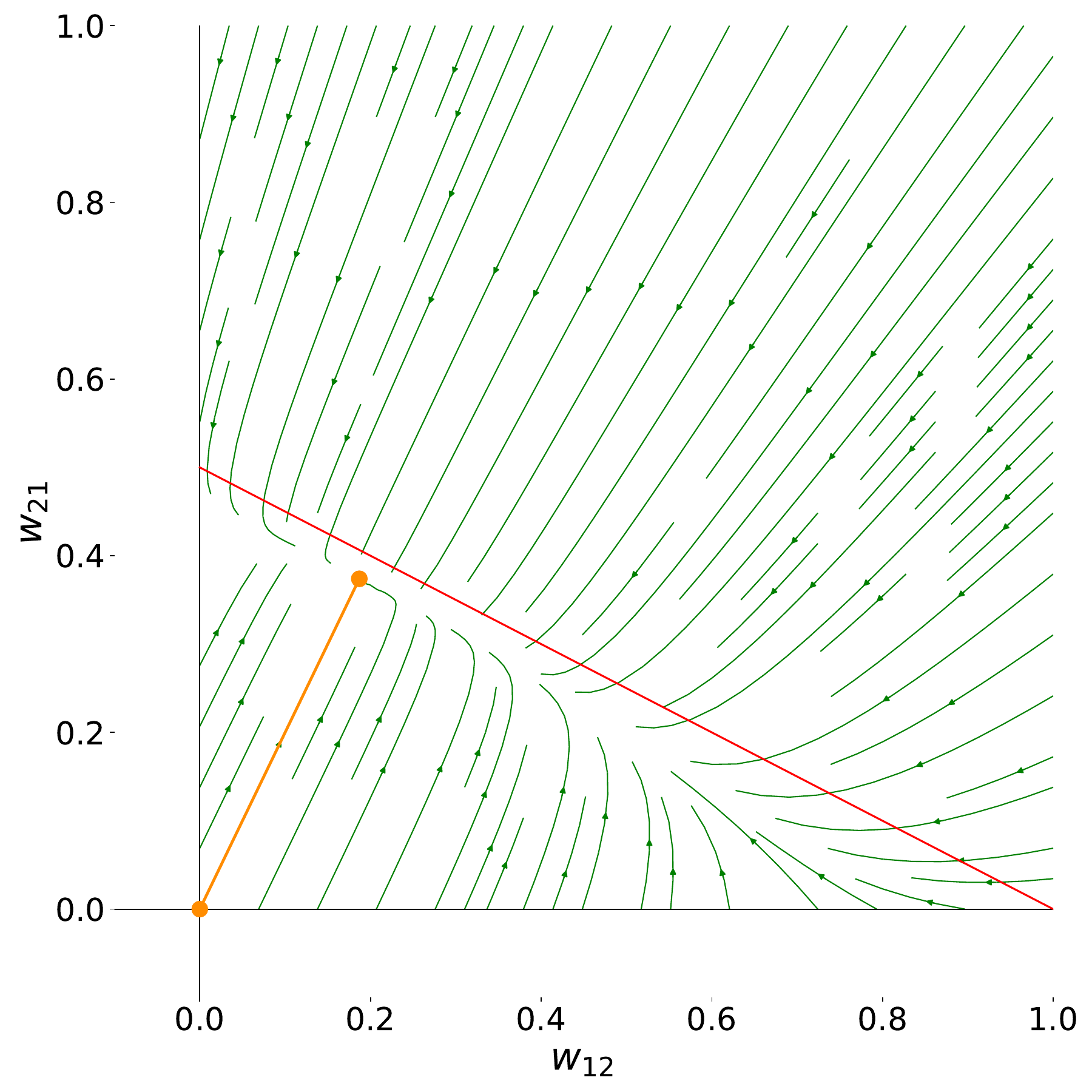}
  \caption{Paths of algorithm (\ref{eq:descentalgo}) for the explicitly regularized loss function $\hat J_2 (w) $ using the $ L^2$-descent, i.e., paths of $ {\cal F}_{ ([0,1]^2 ,2,  \hat J_2  )} ( w) $. It follows that $ {\cal F}_{ ([0,1]^2, 2 , \hat J_2 )} ( 0,0 ) =  (0.187, 0.374)$, denoted by the orange dot. } 
  \label{fig:er2}
\end{minipage}%
\hspace{0.5cm}
\begin{minipage}{.48\textwidth}
  \centering
  \includegraphics[width=\linewidth]{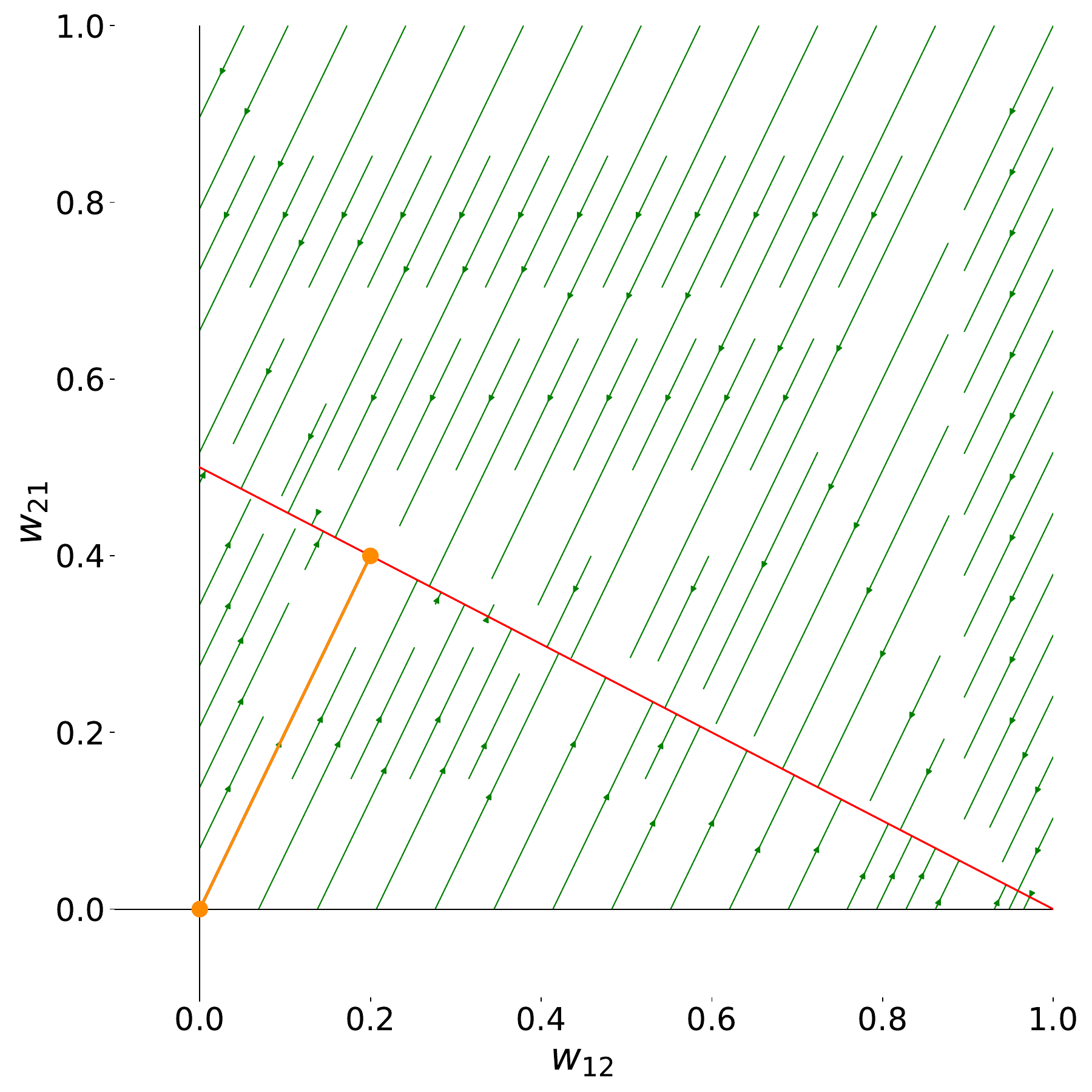}
    \caption{Paths of algorithm (\ref{eq:descentalgo}) for basic loss function $J(w) $ using implicit regularization via the $ L^2$-descent, i.e., paths of $ {\cal F}_{ ( [0,1]^2, 2 , J  ) } ( w ) $. It follows that $ {\cal F}_{ ([0,1]^2, 2 , J_2 )} ( 0,0 ) =  (\frac{1}{5}, \frac{2}{5})$, denoted by the orange dot.}
        \label{fig:ir2}
\end{minipage}%
\vspace{0.1cm}
\begin{minipage}{.48\textwidth}
  \centering
  \includegraphics[width=\linewidth]{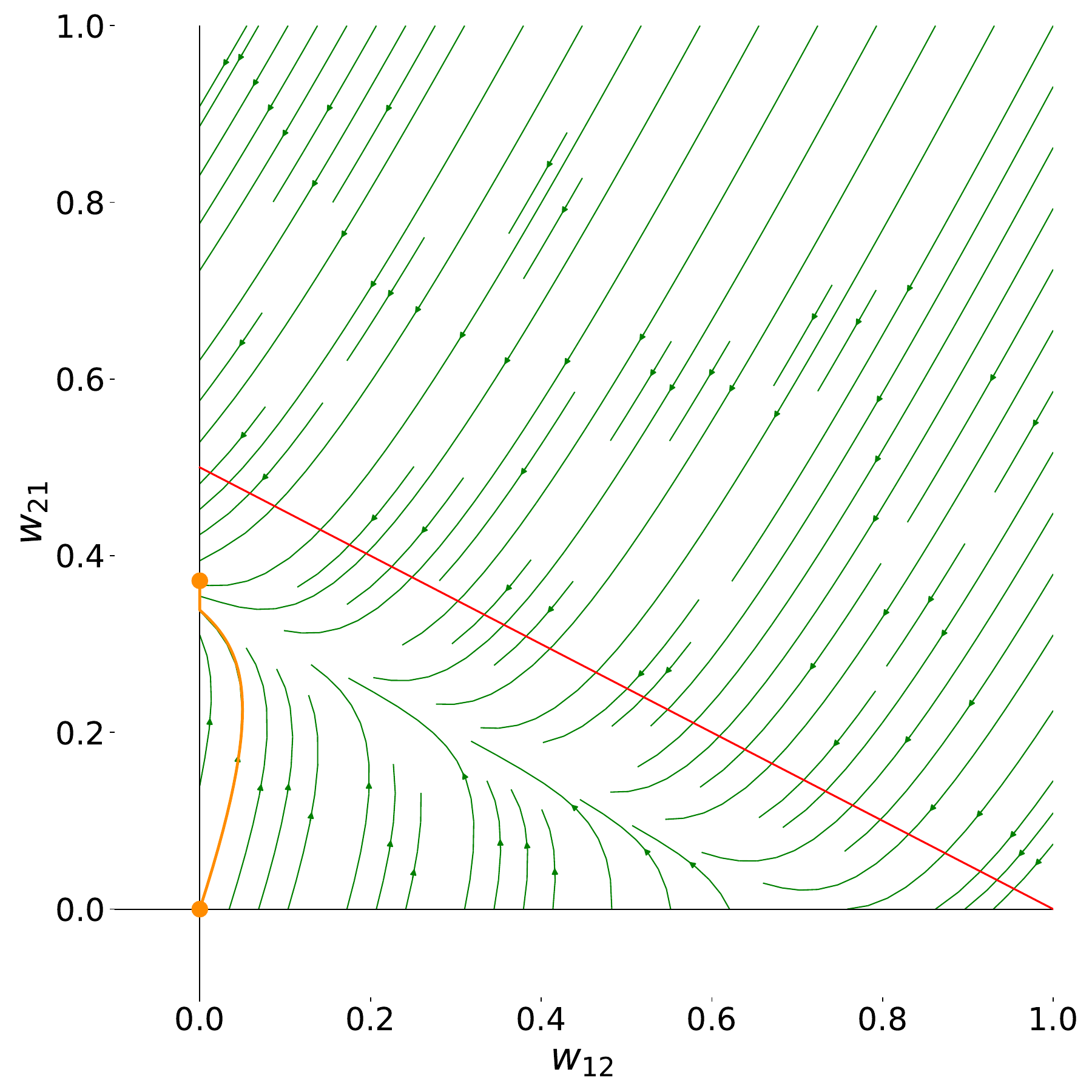}
    \caption{Paths of algorithm (\ref{eq:descentalgo}) for the explicitly regularized loss function  $\hat J_1(w) $ using the $ L^2 $-descent, i.e., paths of $ {\cal F}_{ ( [0,1]^2, 2 ,\hat J_1 ) } ( w ) $.  It follows that $ {\cal F}_{( [0,1]^2, 2 , \hat J_1 )} ( 0,0 ) =  (0, \frac{3}{8})$, denoted by the orange dot.} 
        \label{fig:er1}
\end{minipage}%
\hspace{0.5cm}
\begin{minipage}{.48\textwidth}
  \centering
  \includegraphics[width=\linewidth]{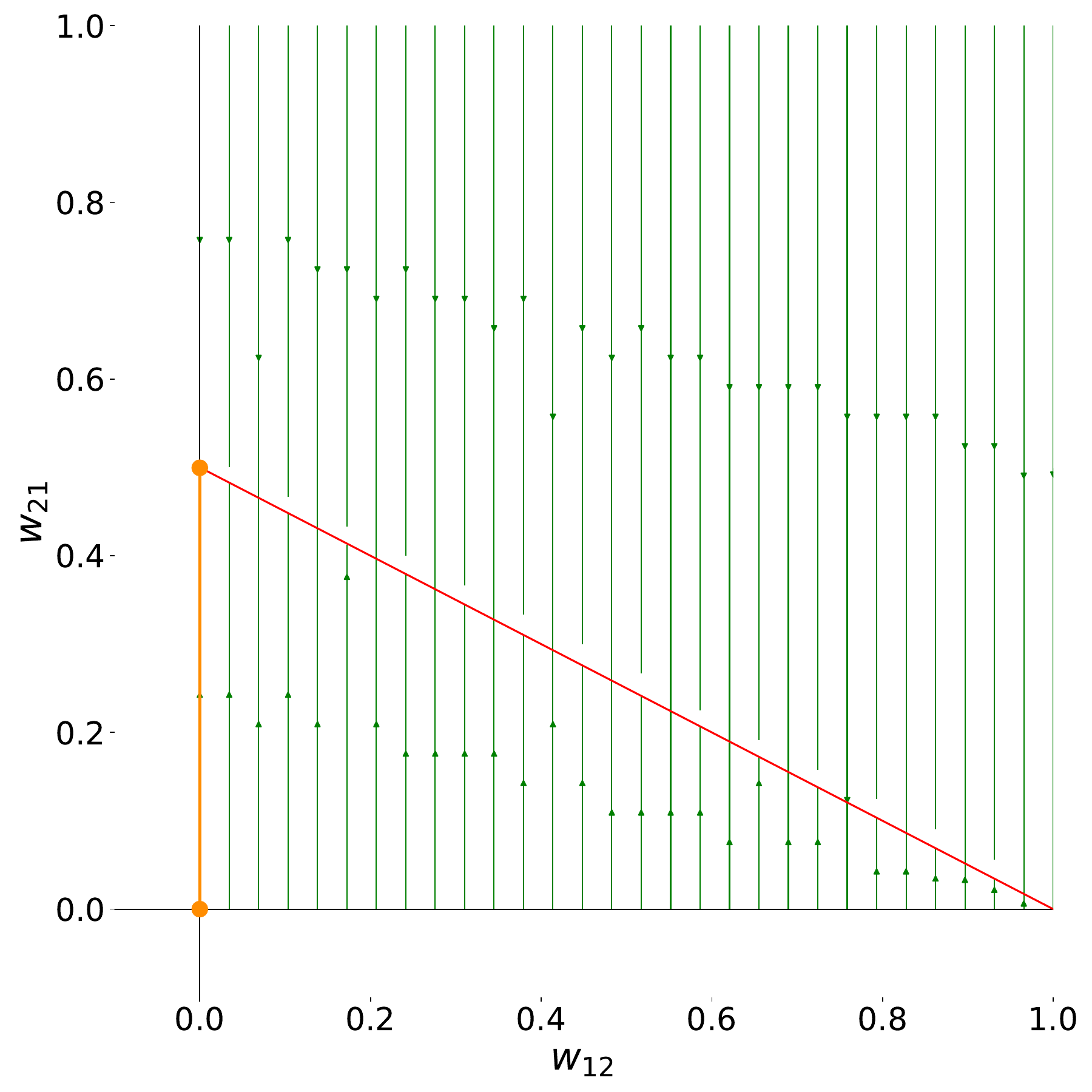}
    \caption{Paths of algorithm (\ref{eq:descentalgo}) for the basic loss function $ J(w) $ using implicit regularization via $ L^1$-descent, i.e., paths of $ {\cal F}_{ ([0,1]^2, 1 , J ) } ( w ) $.  It follows that $ {\cal F}_{ ( [0,1]^2, 1 , J_1 )} ( 0,0 ) =  (0, \frac{1}{2})$, denoted by the orange dot.} 
    \label{fig:ir1}
\end{minipage}%
\end{figure}
In both settings, explicit regularization does prevent the loss function $ J ( w )$ from becoming zero.

An important observation is that the choice of $ p $ has an impact on the trajectories of the algorithm.
Also, we see that in this case the $L^2$-descent in Figure~\ref{fig:ir2} adjusts both components of $ w ( k)$, whereas the $ L^1 $-descent, only adjusts one component.
This result of implicit regularization via the $L^1$-descent can be favorable in fitting graphs to features. We refer the reader to Section~\ref{subsec:social_network_example} for a more elaborate numerical example.}

\end{example}

Typically in feasible direction methods, one aims to find step sizes using line search \citep{zoutendijk1960methods, Rosen_1960}.
Long step sizes may benefit the convergence speed of the algorithm; however, it can harm the implicit regularization of the solution.
This is illustrated in Example \ref{ex:long_short_step}.

\begin{example} \label{ex:long_short_step}
{\em Let us reconsider the set of networks from Example \ref{ex:reg_explicit} with node set $V = \{1,2\}$, 
link set $E = \{(1,2), (2,1)\}$, and link weights $w = ( w_{1,2} ,w_{2,1} )  $ with $ w \in {\cal W}^b = [ 0, 1 ]^2 $.
Now, consider the feature function
\[
\Phi(w ) =  \left( w_{1,2} - 1\right)^2  - w_{2,1} ,
\]
and suppose we want to construct graphs such that $ \Phi ( w ) = 0$.
This leads to the loss function
\[
J ( w ) = \left\| \left( w_{1,2} - 1\right)^2 - w_{2,1} \right\|_2^2.
\]

In Figures \ref{fig:L2_short_vs_long} and \ref{fig:L1_short_vs_long}, we visualize the target set by the red line.
The orange path starting in $w(0) = (0, 0)$ denotes the sequence $w(k)$ induced by the algorithm in (\ref{eq:descentalgo}) using maximal step size $\alpha = 0.001$, i.e., we visualize the output trace of $ {\cal F}_{ ( [0,1]^2, p , J ) } := {\cal F}_{ ( [0,1]^2 ,p,  J ; \alpha , T , \gamma ) } ( w (0)) $, for $ T = 10000 $, where $p=2$ in Figure~\ref{fig:L2_short_vs_long} and $p=1$ in Figure~\ref{fig:L1_short_vs_long}.
Moreover, we visualize the blue path that denotes the sequence $\{\tilde{w}(k)\}$ starting in $\tilde{w}(0) = w(0) = (0, 0)$ induced by line search algorithm Cauchy's steepest descent:
\begin{align}
\begin{array}{rll} \label{eq:cauchy}
     \tilde{w}(k+1) =& \tilde{w}(k) + \alpha(k) \delta(\tilde{w},p)  & k=0,1,2,\dots, T-1\\
     \alpha(k) =& \argmin_{\alpha > 0} J(\tilde{w}(k) + \alpha \delta(\tilde{w}(k),p).
\end{array}
\end{align}

We can observe that the blue path with long step size(s) finds a solution on the line spanned by the initial steepest feasible descent direction (in fact, it is the only iteration of this algorithm).
In contrast, the red path follows a short step size and therefore ensures that we minimize the loss function using a small graph transformation at every step.
Consequently, we find that the distance $\| w(T) - w(0) \|_p$ is smaller than the distance $\| \Tilde{w}(T) - w(0) \|_p$, for both $p \in \{1,2\}$, as can be seen from the black lines that denote all $w$ for which $\|w - w(0)\|_p$.
Hence, the short step size solution is closer to $w(0)$ than the long step size solution.
Note that this small example shows that implicit regularization is not guaranteed to adjust a minimal number of links. 
However, in general, it does lead to a strong reduction in the number of links adjusted, see Example~\ref{ex:reg_explicit} and the numerical example in \ref{subsec:social_network_example}.
}

\begin{figure}
\centering
\begin{minipage}{.48\textwidth}
  \centering
  \includegraphics[width=\linewidth]{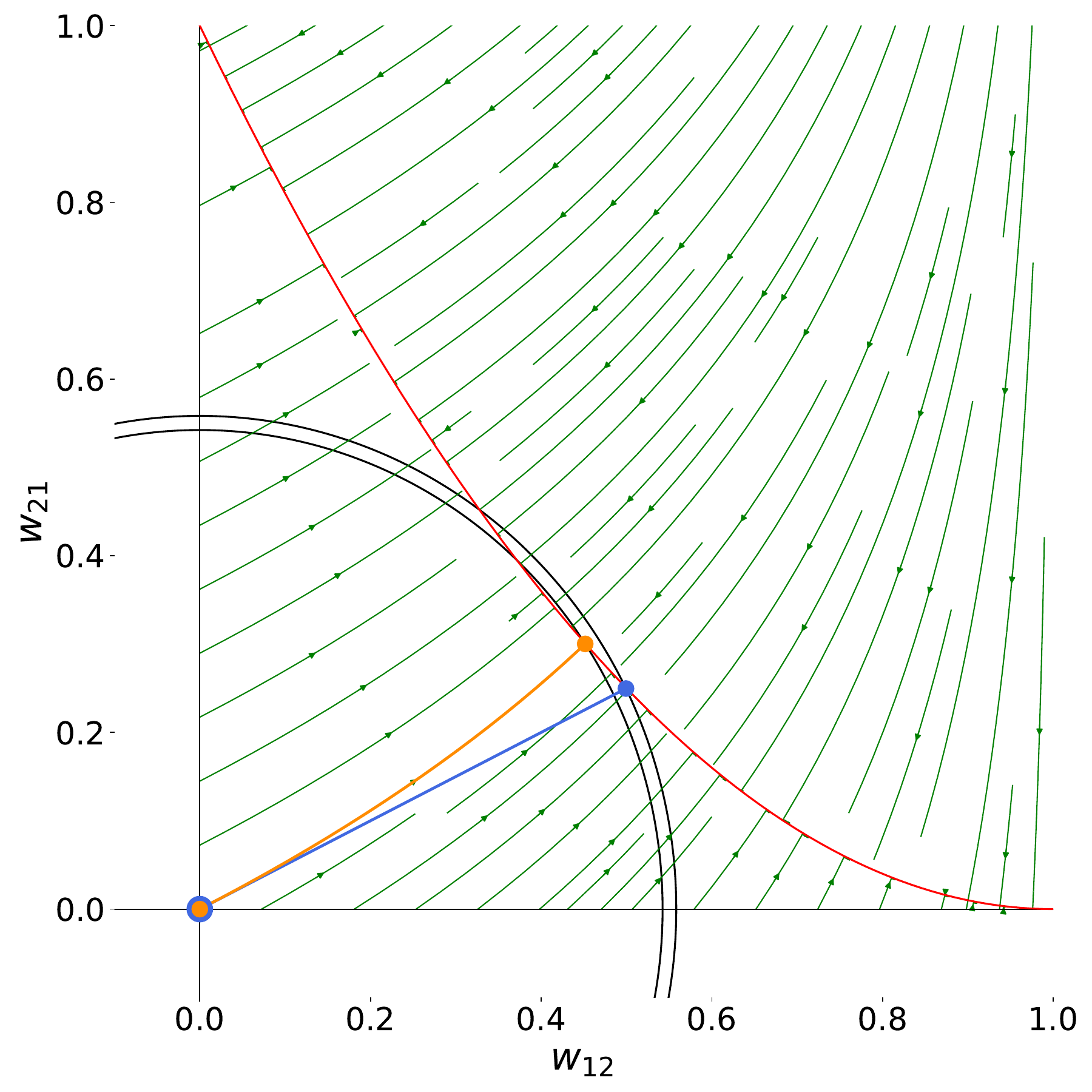}
  \caption{Paths of steepest feasible descent algorithm (\ref{eq:descentalgo}) with short step size (orange) and long step size algorithm \eqref{eq:cauchy} (blue) in the $L^2$-normed vector space $\mathcal{V}_2$. } 
  \label{fig:L2_short_vs_long}
\end{minipage}%
\hspace{0.5cm}
\begin{minipage}{.48\textwidth}
  \centering
  \includegraphics[width=\linewidth]{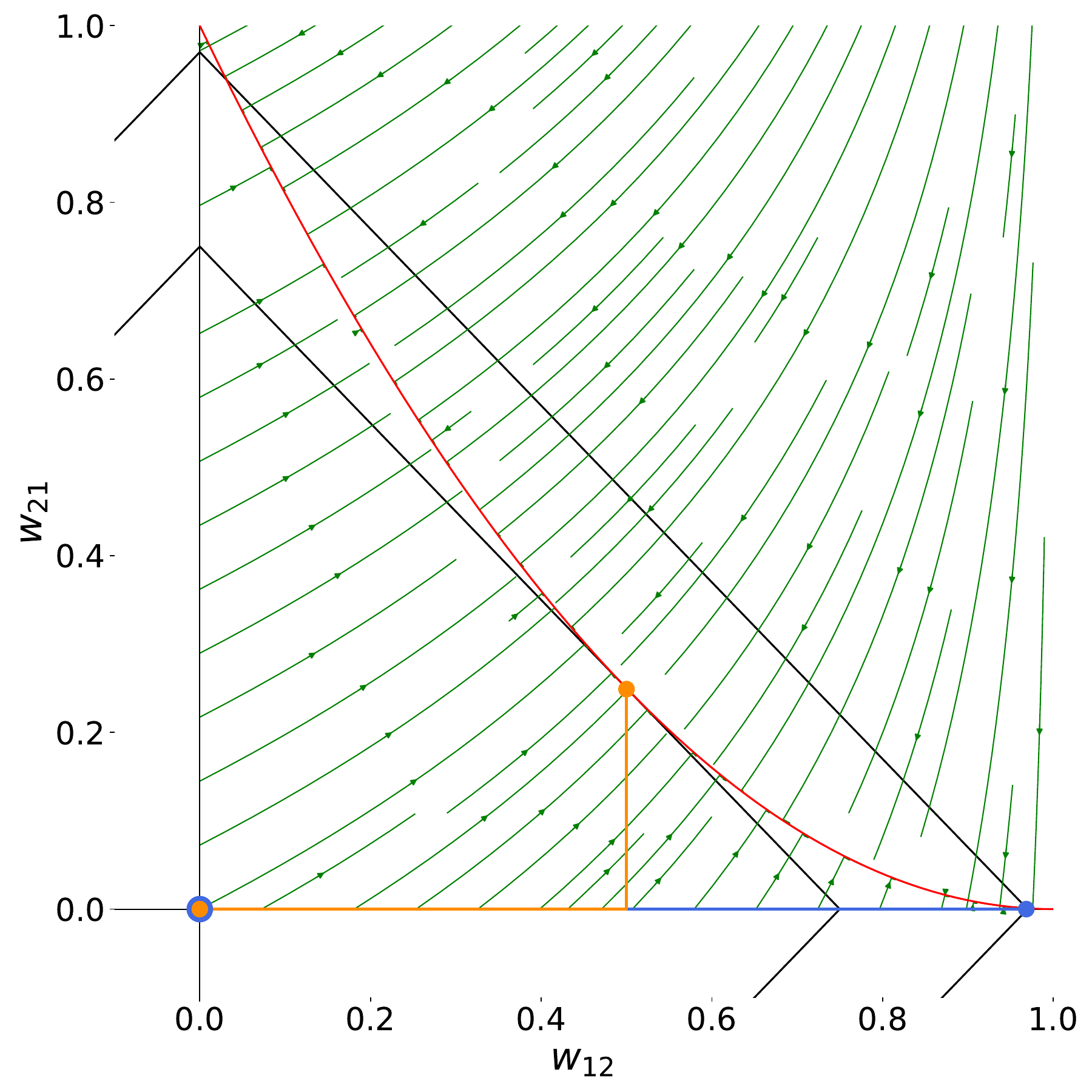}
    \caption{Paths of steepest feasible descent algorithm (\ref{eq:descentalgo}) with short step size (orange) and long step size algorithm \eqref{eq:cauchy} (blue) in the $L^1$-normed vector space $\mathcal{V}_1$.}
    \label{fig:L1_short_vs_long}
\end{minipage}%
\end{figure}
\end{example}

We conclude this section with Remark \ref{rem:gs_hc}, showing a simpler expression for the step size in the case where we consider the $L^1$-descent direction.

\begin{remark} \label{rem:gs_hc}
    {\em Let $\delta(w(k),1)$ be the steepest feasible descent in the $k$th iteration in the bounded setting.  
    The step size can then be calculated as
    \[
     \alpha(k) = \begin{cases}
         \min \big(\alpha, w(k)_{i,j} \big) & \text{if } \delta(w(k),1)_{i,j} = -1; \\
         \min \left(\alpha, \frac{b - w(k)_{i,j}}{\delta(w(k),1)_{i,j}}\right) & \text{if } \delta(w(k),1)_{i,j} = 1 ; \\
         0 & \text{if } \delta(w(k),1)_{i,j} = 0,
     \end{cases}
    \]
    which follows from the fact that there exists at most a single nonnegative element $\delta(w,1)_{i,j} \in \{-1, 1\}$.
    }
\end{remark}

\subsection{Steepest Feasible Descent for Markovian Graphs}\label{sec:MG}

In this section, we present our result for finding the steepest feasible descent direction in the Markov setting. 
In fact, the result is quite similar to the bounded setting.
Recall for the following that we denote by $u_{ i, j }$ a vector of size $|E|$ with value $1$ at the entry corresponding to $ (i, j ) \in E$, and $0$ otherwise.

\begin{theorem}\label{ref:thmdescentMarkov}
Let $w \in {\cal W}^M $, then a steepest feasible descent direction is given
\begin{itemize} 
\item[(i)] for $ p=2$, by $\delta(w,2) = \left(\frac{\Tilde{\delta}(w,2)_{i,j}}{\| \Tilde{\delta}(w,2) \|_2}\right)_{(i,j) \in E}$, where  
\[
\Tilde{\delta}( w , 2 )_{i,j} =
\begin{cases}
 \max \left\{ \lambda_i - \nabla J(w)_{ i,j }, 0 \right\} &(i,j) \in \mathcal{A}; \\
 \lambda_i - \nabla J(w)_{ i,j } & (i,j) \in E \setminus \mathcal{A},
\end{cases}
\]
where $\lambda_i$ solves the equation
\begin{align} \label{eq:piece_lin}
\sum_{(i,j) \in \mathcal{A}} \max \left\{ \lambda_i - \nabla J(w)_{ i,j }, 0 \right\} + \sum_{(i,j) \in E \setminus \mathcal{A}} \left[ \lambda_i - \nabla J(w)_{ i,j } \right]= 0,
\end{align}
for all $i \in V$, and
$$\mathcal{A} = \{ (i,j) \in E : w_{i,j} = 0 \}.$$
The normalization ensures that $\| \delta(w,2) \|_2 = 1$ if $\delta(w,2) \neq 0$.

\item[(ii)]
for $ p =1 $, by 
\[
\delta( w , 1 ) = \frac{1}{2} \big (  u_{ i,k }  - u_{ i,j } \big ) , 
\]
where 
\[
( i,j,k ) \in  \argmax_{i,j,k} \  \left( \delta( w , 2 )_{i,k} - \delta( w , 2 )_{i,j} \right) . 
\]
In the case that $\argmax_{i,j,k} \  \left( \delta( w , 2 )_{i,k} - \delta( w , 2 )_{i,j} \right)$ is not a singleton set, we simply choose $(i,j,k) \in \argmax_{i,j,k} \  \left( \delta( w , 2 )_{i,k} - \delta( w , 2 )_{i,j} \right)$ randomly.
\end{itemize}

\end{theorem}
\begin{proof}
For $p=2$, we find the steepest feasible descent on the simplex for $w_{(i,\cdot)}$, for $i\in V$, and subsequently normalize it to satisfy $\| \delta(w,2) \|_2 = 1$, see \citet{gafni1984two}. 
Note that \eqref{eq:piece_lin} is piecewise linear and therefore efficiently solved.

For $p=1$, the proof is very similar to the proof in Theorem~\ref{ref:thmdescent}.
We again leverage the obtained $\delta( w , 2 )$.
Given the constraint $\| \delta( w , 1 ) \|_1 \leq 1$ and $\sum_{j} \delta( w , 1 )_{i,j} = 0$, for all $i \in V$, we know that a corner solution minimizes $\partial_d J(w) = d \cdot \nabla J(w)$ by the linearity of the problem.
In words, we must find a row $i$, in which a weight exchange between two links maximally decreases the loss function.
More formally, choose the element $(i,k)$ of $\delta( w , 2 )$ with the steepest decrease in loss function $J$, and similarly choose the element $(i,j)$ of $\delta( w , 2 )$ with the steepest increase and let $\delta( w , 1 ) = \frac{1}{2} \big (  u_{ i,k }  - u_{ i,j } \big )$
Note that satisfies $\| \delta( w , 1 ) \|_1 \leq 1$ and the directions in $(i,k)$ and $(i,j)$ are both feasible directions by construction of $\delta( w , 2 )$.
\end{proof}

Finally, we obtain a more specific expression for $\alpha(k)$ if we consider the $L^1$-descent direction, see Remark~\ref{rem:gs_mg}.

\begin{remark} \label{rem:gs_mg}
    {\em Let $\delta(w(k),1)$ be the steepest feasible descent in the $k$th iteration in the bounded setting.  
    The step size can then be calculated as
    \[
     \alpha(k) = \min \{\alpha, 2 w(k)_{i,j}\},
    \]
    which follows from the fact that there is at most one negative element $(i,j)$ for which $\delta(w,1)_{i,j} = -\frac{1}{2}$.
    }
\end{remark}

\subsection{Steepest Feasible Descent through Projection} \label{sec:proj_sfd}

An alternative approach to find $\delta(w,p)$, for $p=2$ specifically, is to utilize a projection on the convex feasible region $\mathcal{W} \in \{\mathcal{W}^b, \mathcal{W}^M\}$.
More formally, let us define
\begin{align}
    \Pi_\mathcal{W}(w) := \argmin_{w' \in \mathcal{W}} \| w - w' \|_2,
\end{align}
Now one can show that in $w$, such that $\delta(w,2) \neq 0$, it holds that for $\alpha \to 0$,
\begin{align}
 \delta^{\alpha}(w,2) := \frac{\Pi_{\mathcal{W}}(w - \alpha \nabla J(w)) -w }{\| \Pi_{\mathcal{W}}(w - \alpha \nabla J(w)) -w \|_2} \to \delta(w,2),
\end{align}
see \citet{zhang1995stability,benveniste2012adaptive}. 
In case $\mathcal{W} = \mathcal{W}^b$, $\Pi_{\mathcal{W}^b}(\cdot)$ is the straightforward projection on the hypercube and in case $\mathcal{W} = \mathcal{W}^M$, $\Pi_{\mathcal{W}^M}(\cdot)$ can be constructed using a projection on the probability simplex of $w(i,\cdot)$, see, e.g., \citet{condat2016fast, perez2020filtered}, for all $i \in V$.

We can use an approximation of $\delta(w,2)$ in a descent algorithm similar to \eqref{eq:descentalgo}.
Let $\Bar{w}(k) = \Pi_{\mathcal{W}}(w(k) - \alpha \nabla J(w(k)))$ and so $(\Bar{w}(k) - w(k))$ is a descent direction (note this direction does not necessarily have unit length).
Now, we can apply the following algorithm:
\begin{align}\label{eq:descentalgo_eta}
   w (k+1 ) &= w ( k ) + \alpha(k) (\Bar{w}(k) - w(k)), \qquad \forall k = 0, 1, 2, \dots, T-1 ,
\end{align}
where $\alpha(k) = \beta^\kappa \alpha$ and $\alpha> 0$ small for regularization purposes, such that $\kappa$ is the smallest non-negative integer satisfying determined following the Armijo rule
\begin{align}
J(w(k)) - J(w(k) + \beta^{\kappa} (\Bar{w}(k) - w(k)) \geq - \sigma \beta^{\kappa} \nabla_{(\Bar{w}(k) - w(k))} J(w(k)).
\end{align}
Convergence can be shown by a verbatim repetition of the proof in Theorem~\ref{thm:convergence} \citep{bertsekas1997nonlinear}.

Some features allow for direct enforcement through projection, possibly through iterative projection on convex sets, such as Dykstra's method \citep{boyle1986method}.
In this case, there is no need to include these features in the loss function and we simply enrich $\mathcal{W}$ to include these features.
For example, we will work with graphs that satisfy in- and out-strengths in Sections \ref{sec:application1} and \ref{sec:application2}. 
In that setting, for given feature functions $s^+(w)$ and $s^-(w)$, we let $\mathcal{W} = \{ w \in \mathcal{W}^b : s^+(w) = s^{+*}, s^-(w) = s^{-*}\}$ and find the projection on $\mathcal{W}$ using Dykstra's method projecting on (scaled) simplices iteratively to enforce the features $s^+(w) = s^{+*}$ and $s^-(w) = s^{-*}$.

Moreover, projected gradient descent allows for working with features functions that do not immediately satisfy continuous differentiability {\bf (A1)} and {\bf (A2)} for all $w \in \mathcal{W}$.
For example, assortativity (Appendix~\ref{sf:assortativity}) does not satisfy {\bf (A1)} and {\bf (A2)}.
However, when assuming non-identical in-strengths $s^{+*}$, assortativity does satisfy continuous differentiability for all $w \in \mathcal{W} = \{ w \in \mathcal{W}^b : s^-(w) = s^{+*}\}$.

\section{Application: Hard Constraint Sampling of Networks} \label{sec:application1} 

In this section, we explain how the FBNC algorithm can be leveraged for network sampling. 
We argue that our method is of particular interest when obtaining a maximum entropy distribution of networks is not required. 
A particular benefit of our approach is that we force the graph to satisfy features exactly (as opposed to weak constraint sampling which is predominant in the literature) while allowing for a wide variety of features. 
To illustrate, we provide a numerical example in which there is a need to reconstruct a network satisfying partially available data. More generally, if one needs an ensemble of networks satisfying a specific set of features for testing graph-based algorithms, such as clustering algorithms for networks, the FBNC algorithm can serve as a random network generator.

\subsection{Literature Review}
Sampling networks satisfying constraints is a well-studied problem with a strong foundation in the field of statistical physics.
Many approaches aim for sampling of networks with maximum entropy in the distribution of sampled networks.
This is of particular interest in the case of graph reconstruction, where the distribution of graphs is designed to incorporate the observed information data while being as uncertain as possible with respect to other properties of the graph.
Predominant in this literature is the focus on {\em weak constraint sampling}, where ``weak constraints'' refer to the fact that the sampled networks satisfy the desired feature(s) only on average and the ensemble of networks to be sampled from is described as (macro)-canonical.

One of the most well-known examples of weak constraint sampling is the Erd\H{o}s-Rényi Model (ER-Mode) introduced by \citet{erdos1959random}.
In this model, the presence of a link between any two possible nodes in the network is independent and identically distributed by a Bernoulli random variable with probability $p$.
It follows that a network sample constructed by this model yields an expected density of the number of nodes $N$ multiplied by $p$.
Next to the network density, the degree distribution (or sequence) is perhaps one of the simplest examples of a features that occurs in data-driven studies of real-world networks, and research on random network models satisfying a degree distribution is abundant \citep{barabasi1999emergence,britton2006generating, newman2009random, chen2013directed}. 
A particular statistical model in weak constraint network sampling is the Exponential Random Graphs Model (ERGM) \citep{holland1981exponential,snijders2002markov,handcock2010modeling}, allowing for weak constraint sampling satisfying many different features.
We postpone the discussion of this model and its relation with our work to Section~\ref{sec:ERG}.
In these statistical approaches for network sampling, such as the Erd\H{o}s-Rényi Model and the Exponential Random Graph Models, one is interested in the probabilities for obtaining a particular network. This allows for estimating expected values of metrics of random ensembles of networks through sampling, even though none of the sampled networks may actually satisfy features, see Section~\ref{sec:ERG} later on.

Our FBNC approach belongs to the family of {\em hard constraint sampling} approaches, where ``hard constraints'' ensure that sampled networks are guaranteed to satisfy the given features (if possible), for which the ensemble is described as micro-canonical in the statistical physics literature.
A generic statistical model for sampling from the micro-canonical ensemble like the EGRM does for the canonical ensemble does not exist. 
Most research effort concentrates on sampling unweighted networks satisfying prespecified degree sequences \cite{kleitman1973algorithms,mckay1990uniform, artzy2005generating, fosdick2018configuring,bassler2015exact}. 
For example, \citet{bassler2015exact} show how binary networks satisfying a degree sequence and degree assortativity can be sampled.
Another strain of methods uses a Bayesian approach, where prior distributions of links and their weights are assumed. 
For example, \citet{gandy2017bayesian} show how to sample continuously weighted (financial) networks satisfying in- and out-strength sequences.
\citet{glasserman2023maximum} describe an algorithm for sampling weighted binary bipartite and directed networks satisfying degree sequences.
Finally, \citet{Squartini_Caldarelli_Cimini_Gabrielli_Garlaschelli_2018} provide a good overview of, among other things, sampling weighted networks satisfying degree sequences and in- and out-strenghts.
We extend upon this literature with the introduction of the FBNC sampling algorithm, allowing for network sampling satisfying any feature function that is continuously differentiable in the link weights of the network, see {\bf (A1)} and {\bf (A2)}.

\subsection{The FBNC Sampling Algorithm} \label{sec:alg}

We obtain $n$ samples in $\mathcal{W}_{\Phi,\phi^*}$ by taking randomly generated networks $  w^l $, $l =1, \ldots, L$. 
Assuming $E$ is known (although it may also be sampled from, for example, the ER-model), let $ w^l$ be uniformly from $\mathcal{W}$.
Then, we solve the problem in \eqref{eq:pf_theta} using the mapping $ {\cal F}_{ (\mathcal{W}, p , J )} $, which maps $ w^l $ on $\mathcal{W}_{\Phi,\phi^*}$, with $ J $ given as in (\ref{eq:lossf}).

In Example~\ref{ex:3} below, we apply the above sampling methodology to a setting similar to Example~\ref{ex:reg_explicit}.
Due to the simplicity of the example, we are able to derive a probability distribution over $\mathcal{W}_{\Phi,\phi^*}$.
Indeed, the distribution on $\mathcal{W}_{\Phi,\phi^*}$ is a symmetric triangular distribution.
Noting that in the hard-constraint sampling setting, the uniform distribution is the maximum entropy distribution, our example quickly shows the distribution of graphs does not satisfy maximal entropy.
We note that for such a simple instance the distribution of initial networks may be successfully adjusted to construct uniform samples from $\mathcal{W}_{\Phi,\phi^*}$; however, this becomes practically infeasible for more complex target sets.

\begin{example} \label{ex:3} 
{\em
Revisit the setting put forward in Example~\ref{ex:reg_explicit} with $ V = \{1 ,2\}$, $ E = \{ ( 1 ,2 ) , ( 2 ,1 ) \}$, and assume link weights $w = ( w_{1,2} ,w_{2,1} )  $ with $ w \in {\cal W}^b = [ 0, 1 ]^2$.
Now, let
\[
\Phi ( w ) = w_{ 1,2} + w_{ 2,1},
\]
and suppose we want to construct graphs $w \in {\cal W}^b $ such that $ \Phi ( w ) = 1 = \phi^*$. This gives the loss function 
\[
J ( w ) = \|  w_{1,2} + w_{2,1} - 1\|^2_2 ,
\]
Figure~\ref{fig:non-unif} shows $ {\cal W}^b $ and indicates the target set $\mathcal{W}_{\Phi,\phi^*}$ by the red line. Figure~\ref{fig:non-unif} shows the trajectories of $ {\cal F}_{ ([ 0 , 1 ]^2, 2 , J )} ( w ) $ for $w \in \mathcal{W}^b$. For  $ w $ uniformly sampled on $[ 0 , 1 ]^2$, the limiting points of $ {\cal F}_{ ([ 0 , 1 ]^2, 2 , J )} ( w ) $ follow a triangular distribution on the red line as shown in Figure~\ref{fig:non-unif}.
This illustrates that the FBNC sampling algorithm does not construct samples uniformly on $\mathcal{W}_{\Phi,\phi^*}$.
}
\end{example}

\begin{figure}
  \begin{subfigure}[c]{0.45\textwidth}
    \centering
    \includegraphics[width=\linewidth]{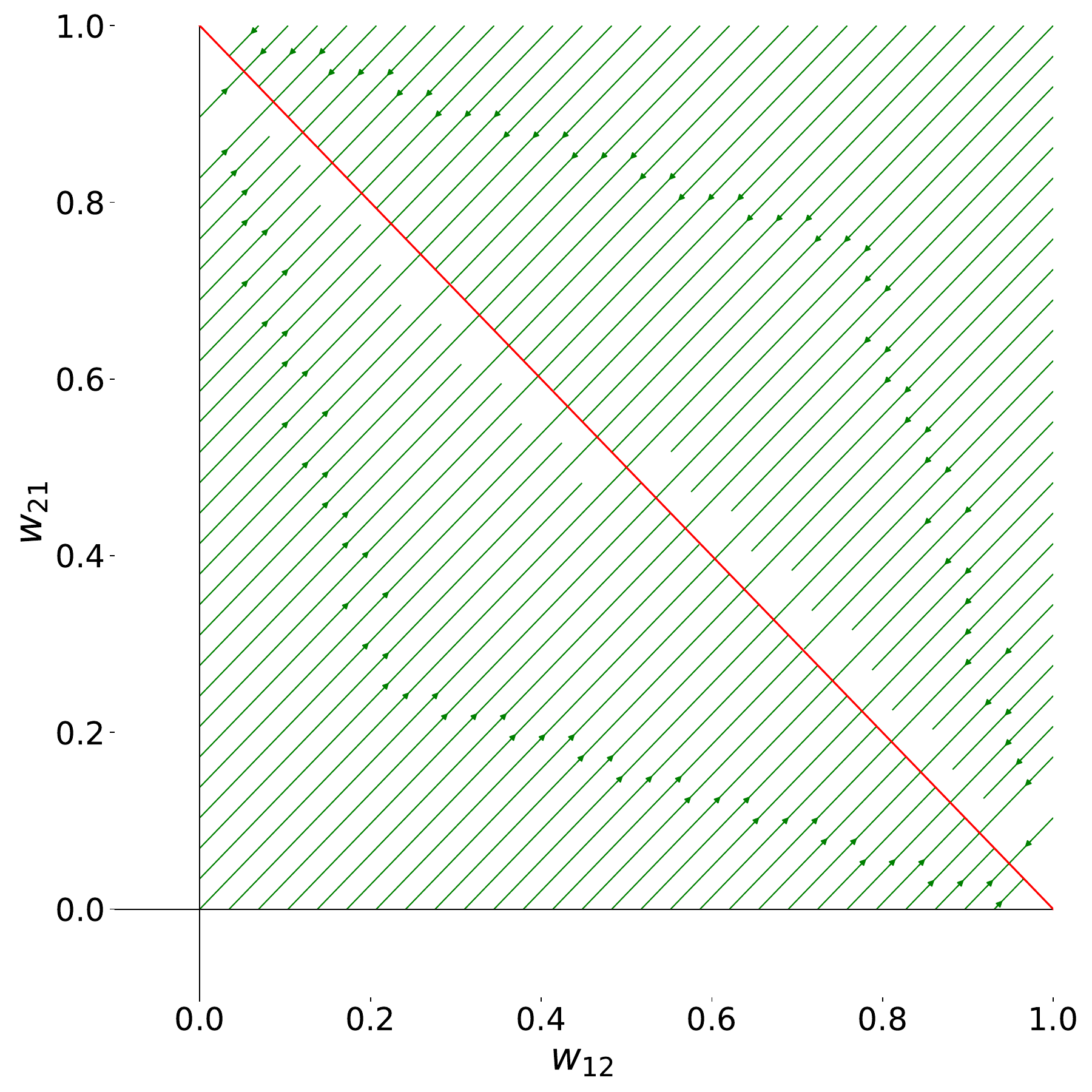}
  \end{subfigure}
  \hfill
  \begin{subfigure}[c]{0.075\textwidth}
    \centering
    \includegraphics[width=0.75\linewidth]{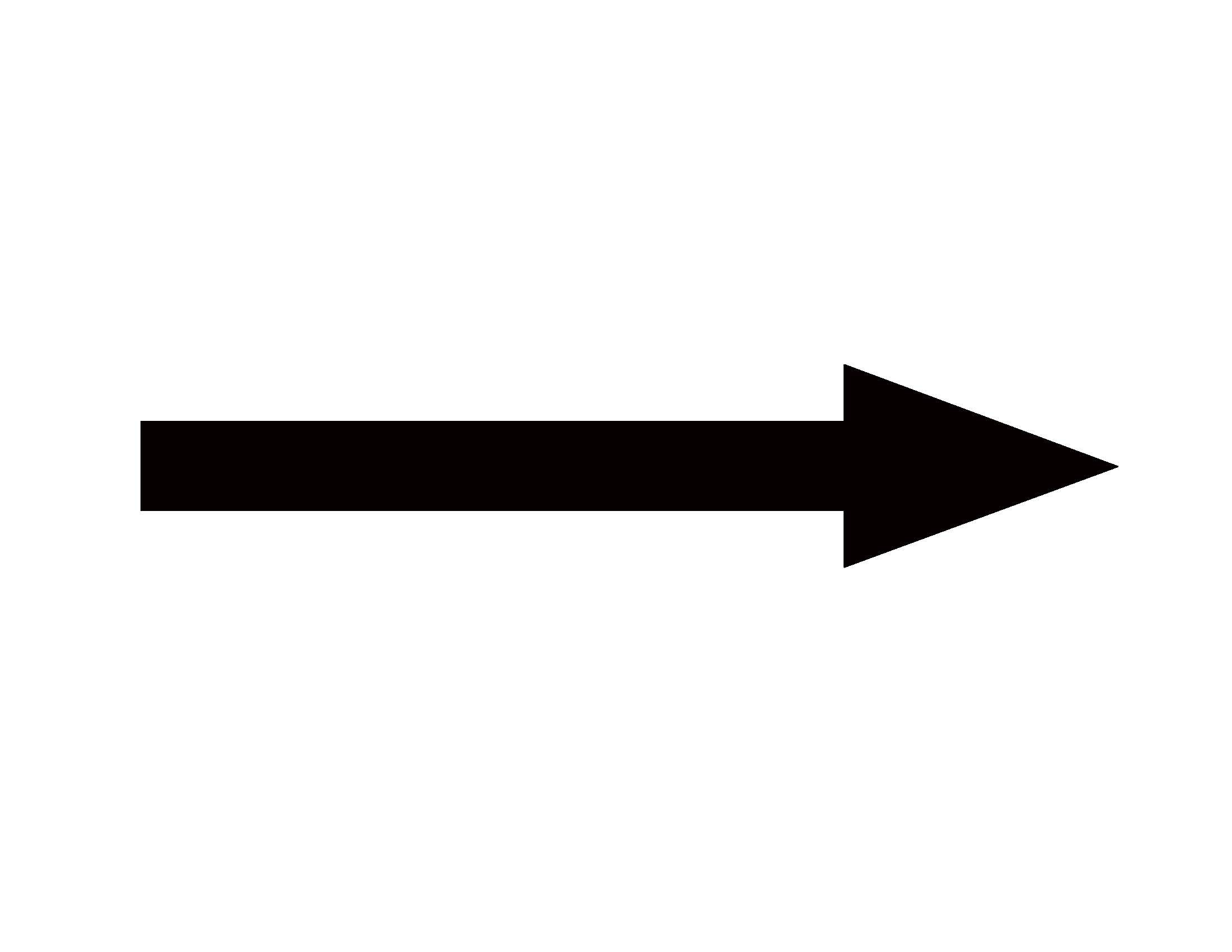}
  \end{subfigure}
  \hfill
  \begin{subfigure}[c]{0.45\textwidth}
    \centering
    \includegraphics[width=\linewidth]{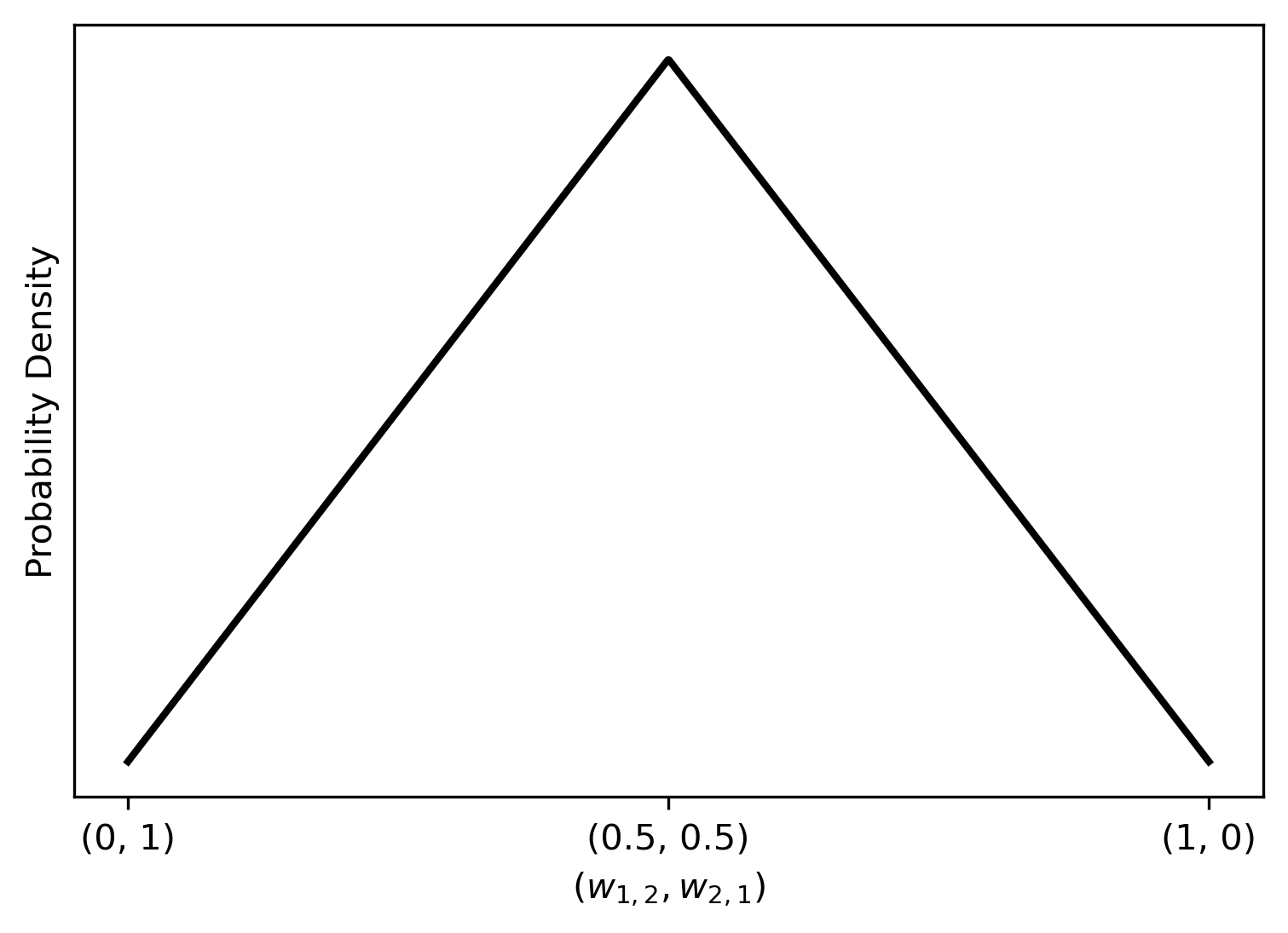} 
  \end{subfigure}
  \vfill 
  \caption{The steepest feasible descent vector field for $J(w)$. Sampling using a uniform distribution of initial networks from $\mathcal{W}^b$ results in an (almost) triangular distribution from $\mathcal{W}_{\Phi,\phi^*}$ (red line).}
\label{fig:non-unif}
\end{figure}

\subsection{Numerical Example: Reconstruction of Confidential Networks} \label{sec:num_example_sampling}

In this numerical example, we demonstrate constructing networks satisfying features using the FBNC sampling algorithm for network reconstruction.
Let us consider a (real-world) partially available financial interbank network of $N = 8$ banks in Argentina during 2018 \citep{forte2020network, github2020Glancszpigel}.
We let $W_{i , j }$ reflect the money lent from bank $i$ to bank $j$. Furthermore, we assume that information is available on the total amount of money lent by each bank, i.e., we suppose that the out-strengths per bank are available and given by
\[
 s^{+*} =  (26.5, 24.9, 30.1, 36.9, 11.1, 24.1, 34.4, 6.3 ),
\]
see (\ref{eq:s+}) for a definition, as well as the total amount of money borrowed by each bank, i.e., we suppose that in-strengths per bank are available and given by
\[
s^{-*} = (43.8, 35.9, 36.5, 14.6, 26.2, 0.9, 35.6, 0.8 ).
\]
To enrich our example, we assume a feature indicating the (relative) dispersion of weights based on the Herfindahl–Hirschman index, which we will refer to as the \textit{concentration index}, given by 
\begin{align} \label{eq:HOSglobal}
    H(W) = \sum_{(i,j) \in E} \left( \frac{W_{i,j}}{s_i^+}\right)^2,
\end{align}
and its feature value is
\[
H^* = 2.84 .
\]
Moreover, note that graphicality of the three features is automatically satisfied as the features coexist within the partially available network.

If set of links $E$ is unavailable, we may simply assume that $E = V^2 \setminus \{ (i, i) : i \in V\}$; however, this may cause network samples to become overly dense.
Another approach is to use an appropriate random graph model to sample the support of the constructed graphs, such as the Erd\H{o}s-Rényi model, where we estimate the density of the graph or assume the graph to be sufficiently dense for the features to be graphical.
More informative approaches for choosing $E$ include the in- and out-degree sequence and the in- and out-strength sequence.
We refer the reader to \citet{Squartini_Caldarelli_Cimini_Gabrielli_Garlaschelli_2018} for an overview of such methods used in network reconstruction. 
Typical in these approaches is the focus on maximum entropy sampling of networks satisfying node degrees and/or strengths in expectation or exactly.
Our approach differs in the fact that sampled networks are from an unknown distribution; however, does allow for sampling of networks satisfying features beyond node degrees and strengths in the hard constraint setting, such as the concentration index.

\begin{figure}
    \centering
    \includegraphics[width=0.45\textwidth]{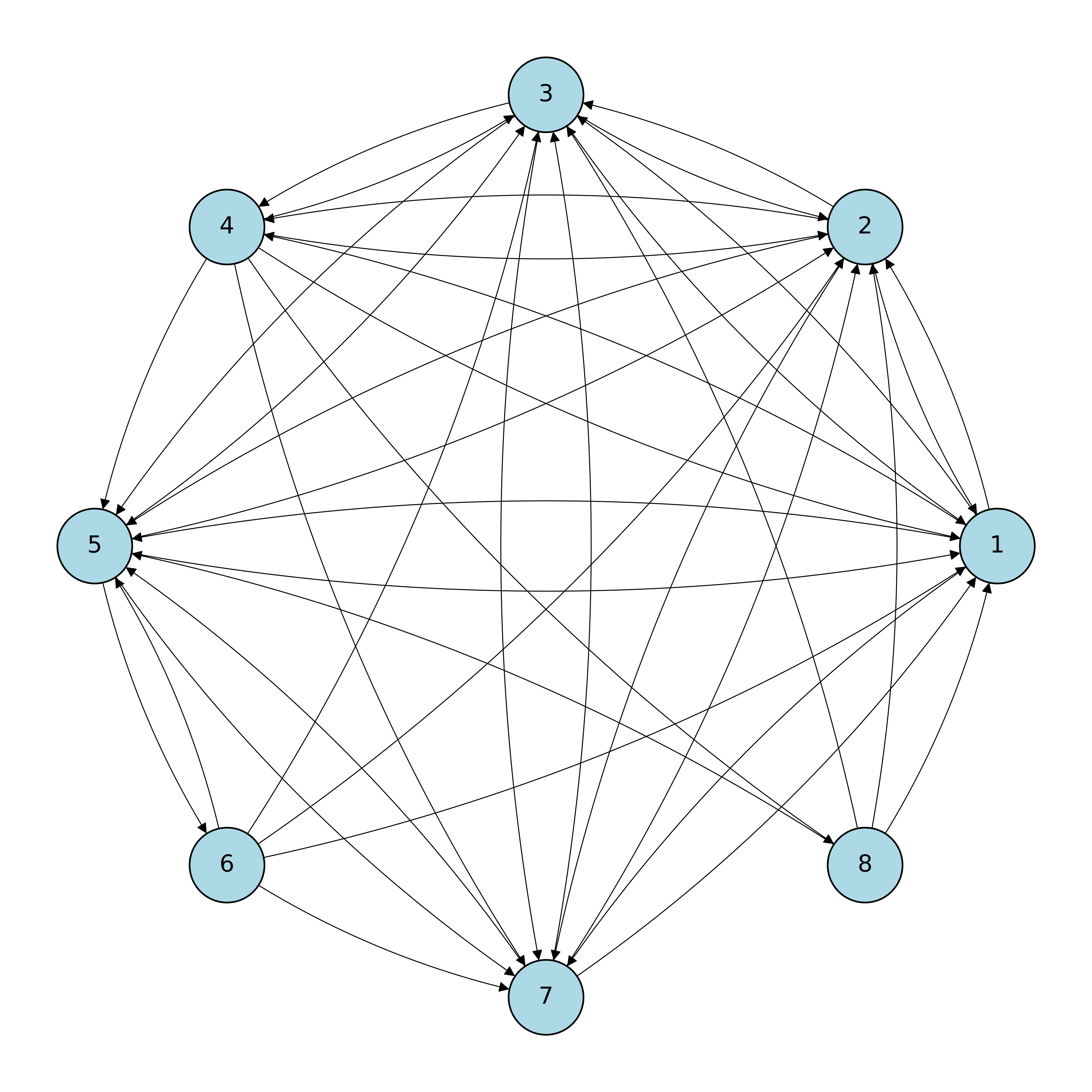}
    \caption{Graph $(V,E)$ used for the Argentinian interbank network reconstruction. The density of the network equals 0.45.}
    \label{fig:supp}
\end{figure}

In the current example, we will assume we can observe the set of links $E$ of the Argentinian financial network, as shown in Figure~\ref{fig:supp}.
We let 
\[
\mathcal{W}^b = \left[ 0,\min \left( \max_i s^{+*}_i, \max_i s^{-*}_i \right) \right]^{|E|} = \left[ 0,36.9\right]^{|E|}
\]
and we can directly project a $w \in \mathcal{W}^b$ to the set
\[
W = \left\{ w \in \mathcal{W}^b : s^+(w) = s^{+*}, s^-(w) = s^{-*} \right\}
\]
using Dykstra's projection method \citep{boyle1986method}.
Now, we can apply the algorithm in \eqref{eq:descentalgo} with the loss function
\[
J ( w ) = \|H(w) - H^*\|^2_2,
\]
to construct $L$ networks via 
$$
\hat{W}_l = \mathcal{F}_{({\cal W}, 2, J; \alpha, T, \gamma )}(w^l), \quad \forall l = 1, \ldots, L,
$$ 
where $w^1, \dots, w^{L}$ are sampled from $\mathcal{ W }$ through projecting uniform samples from $\mathcal{ W }^b$ on $\mathcal{ W }$.
For the present example, we let $L=1000$ and choose the values $\alpha = 20.0$ and $\sigma = 0.001$ as in the current setting, we are not necessarily interested in implicit regularization.
Additionally, we set $T = 10000$ and tolerance parameter $\gamma = 0.001$.

\begin{figure}
\begin{minipage}{.48\textwidth}
  \centering
  \includegraphics[width=\linewidth]{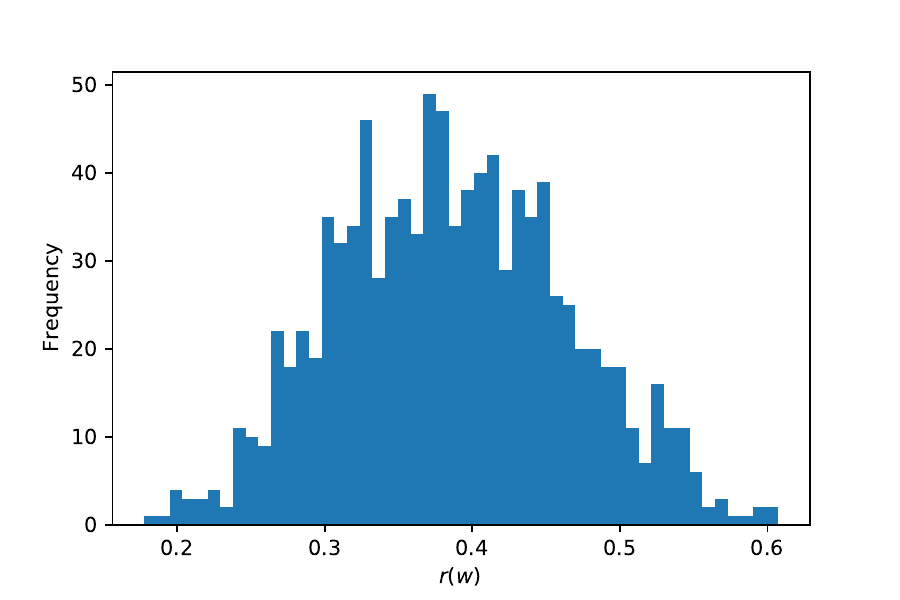}
    \caption{Histogram of reciprocity values for all $\hat{w}_l$, for $l = 1, \ldots, 1000$.} 
    \label{fig:R}
\end{minipage}%
\hspace{0.1cm}
\begin{minipage}{.48\textwidth}
  \centering
  \includegraphics[width=\linewidth]{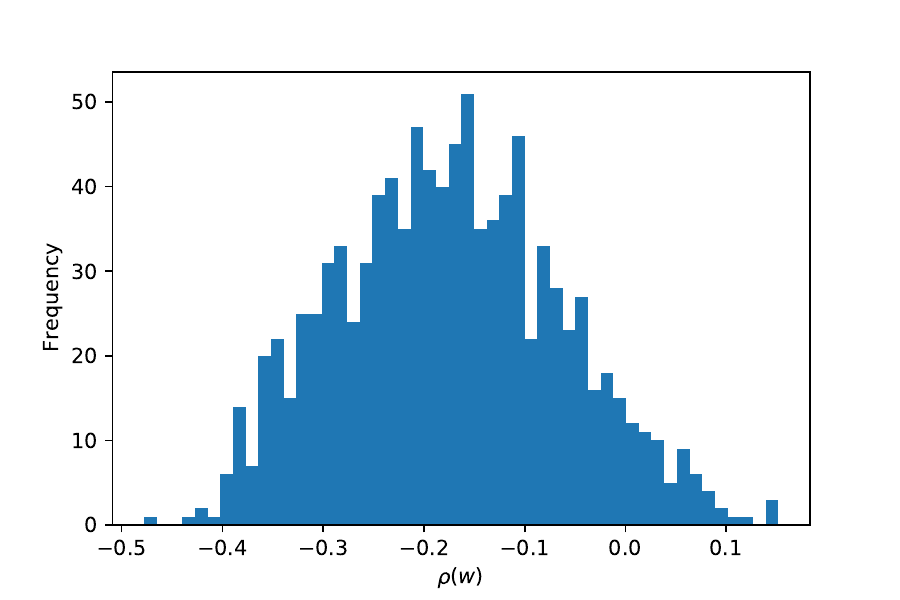}
    \caption{Histogram of assortativity values for all $\hat{w}_l$ for $l = 1, \ldots, 1000$.} 
    \label{fig:rho}
\end{minipage}%
\end{figure}

In the experiment, all network samples $\hat W_l$, for $l = 1, \ldots, L$, had the requested feature values. In Figures \ref{fig:R} and \ref{fig:rho}, we show the diversity of the constructed networks by computing two alternative features: reciprocity $r(w)$ (see Appendix \ref{sf:reciprocity}) and assortativity $\rho(w)$ (see Appendix \ref{sf:assortativity}) for the constructed networks.
Furthermore, we visualize two samples in Figures \ref{fig:high_re} and \ref{fig:low_re}.
As can be observed in Figures \ref{fig:R} and \ref{fig:rho}, there is a wide variety of graphs that satisfy the same link set $E$, in-strengths, out-strengths, and the concentration index. The key takeaway from the example is that even for a given link set $ E $ and three features (in- and out-strengths, as well as given concentration index), the target set is rich, and randomly started the FBNC algorithm constructs a wide variety of samples.

\begin{figure}[!h]
\begin{minipage}{.45\textwidth}
  \centering
  \includegraphics[width=\linewidth]{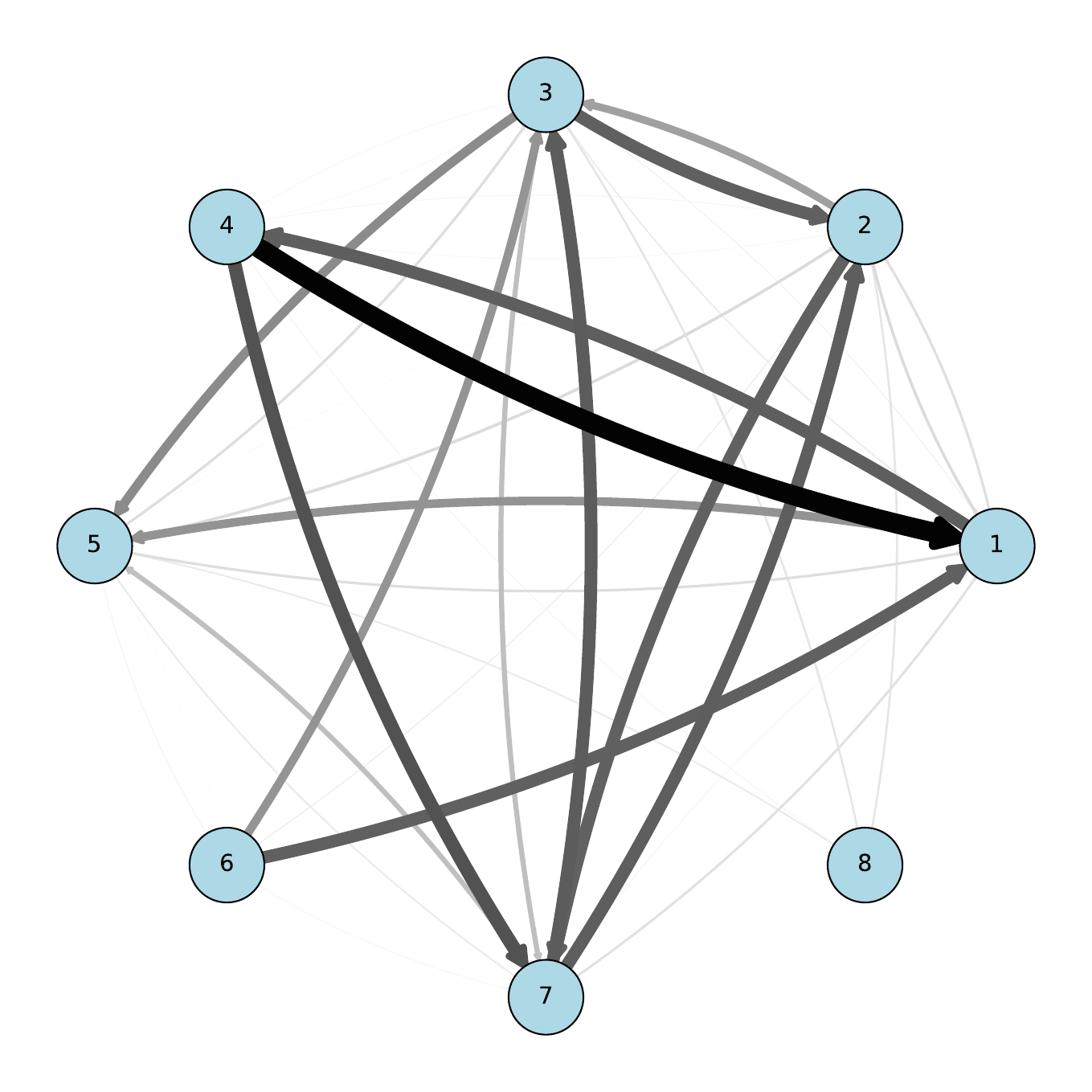}
    \caption{Network $\hat{W}_{164}$. Thicker and darker arrows correspond to larger link weights.}  \label{fig:high_re}
\end{minipage}%
\hspace{0.04\textwidth}
\begin{minipage}{.45\textwidth}
  \centering
  \includegraphics[width=\linewidth]{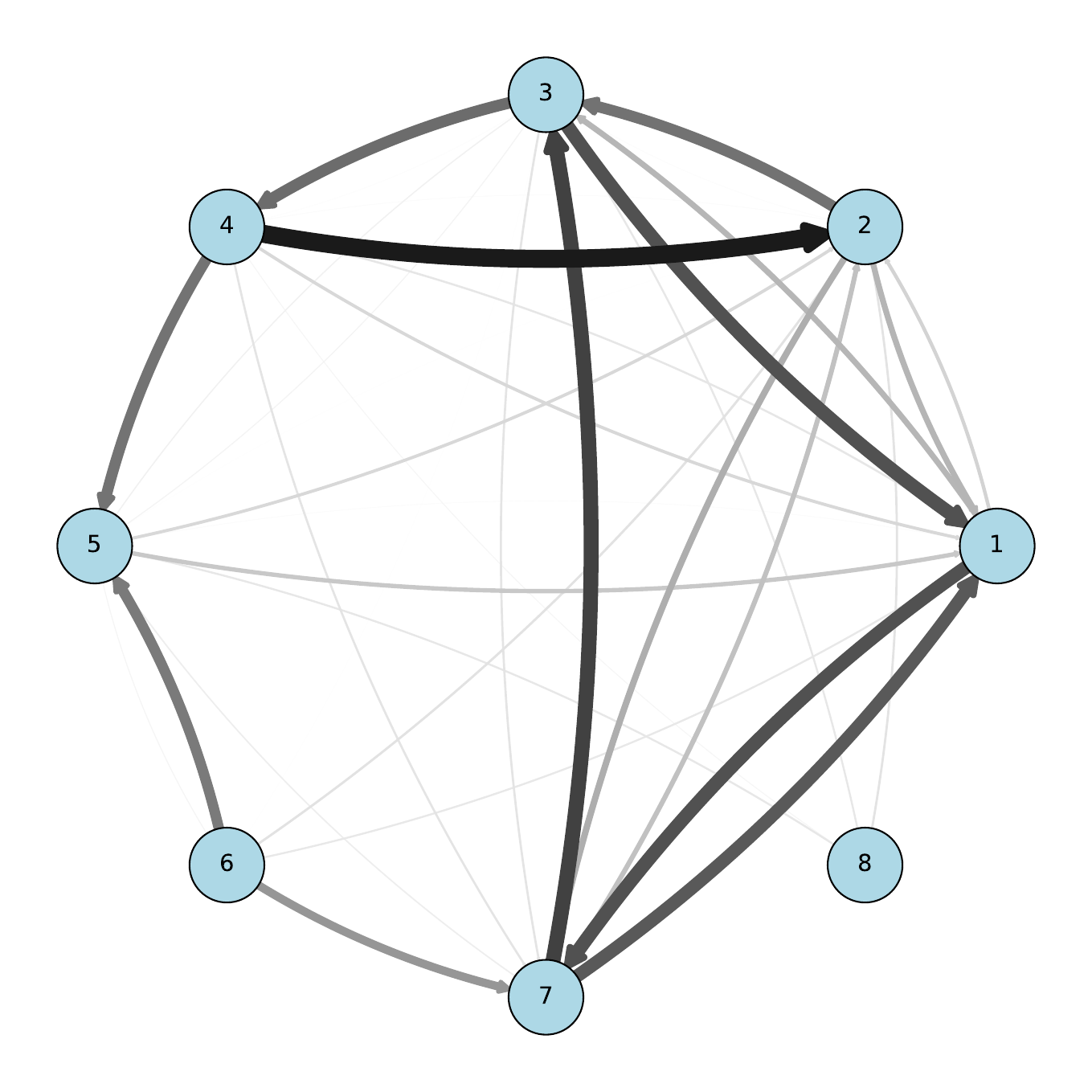}
    \caption{Network $\hat{W}_{516}$. Thicker and darker arrows correspond to larger link weights.} \label{fig:low_re}
\end{minipage}%
\end{figure}

\subsection{Relation to Exponential Random Graph Models} \label{sec:ERG}

Frequently used models for weak constraint network sampling are Exponential Random Graph Models (ERGM) \citep{holland1981exponential,snijders2002markov,handcock2010modeling}.
ERGMs are mostly used for modeling unweighted networks, specifically in social network analysis \citep{lusher2013exponential}. 
The ERGM attempts to model how links are formed in the network. 
To that end, a maximal entropy probability distribution is constructed over all networks such that the observed network features are reproduced in expectation.
A particular strength of ERGMs is that any structural feature can be used. 
Therefore, the ER-Model can be considered a special case for the ERGM. 
The flexibility of features used allows us the ERGM to model a vast array of networks, including both directed and undirected networks, and finally, the specification of the ensemble model can be transformed onto the support of the space of weighted networks, introducing the Generalized Exponential Random Graph Model (GERMG) \citep{desmarais2012statistical,krivitsky2012exponential}.

To obtain network samples using ERGMs, we would first need to construct the ensemble model. The parameters of the ensemble model are estimated by Maximum Likelihood Estimation, which is usually approximated using a Markov Chain Monte Carlo approach due to the inherent intractability, see \citet{snijders2002markov}. 
Then, we can sample networks from this ensemble model using a similar Markov Chain Monte Carlo approach. 

Although weak sampling methods are powerful in explaining relationships in networks, they are not directly applicable to constructing networks from the micro-canonical ensemble in case of continuously weighted features.
Such ensembles are defined by hard constraints that span the target set in a subspace.
Hence, the probability for sampling graphs strictly satisfying a continuously weighted feature up to a \textit{tolerance parameter} $\gamma$, that is, $|\Phi(w) - \phi^*| \leq \gamma$, tends to 0 as $\gamma \to 0$.

\section{Application: What-If Analysis} \label{sec:application2}\label{sec:fbns}\label{sec:what-if}

In this section, we show how the FBNC algorithm can be used for what-if analyses to get insight into the impact of stimulating a certain feature (value) in an existing network, for example, a feature that expresses network connectivity.
This type of analysis can help managing desirable features in the network structure. 
Moreover, the resulting steepest descent trajectory can be a road map to realize an additional feature value in an existing network, thus providing a useful tool in policy development.

Section~\ref{subsec:social_network_example} provides a social network example in which network connectivity is increased using the FBNC algorithm. Section~\ref{subsec:financial_network_example} demonstrates how financial exposures can be efficiently diversified in financial networks, aiming to reduce potential systemic risk.

\subsection{A Social Network Example}\label{subsec:social_network_example}

In this section, for improving its connectivity using a small number of additional links.
Let us consider a (fictive) organizational network $ W_0 $ of $N=27$ agents as shown in Figure~\ref{fig:school_class}. 
Matrix $W_0 \in \mathcal{W}^M$ is stochastic, so that link weights reflect the relative strength of an advice-seeking relationship between a pair of agents.
Colors indicate clusters of agents.

\begin{figure}
    \centering
    \includegraphics[width=0.49\linewidth]{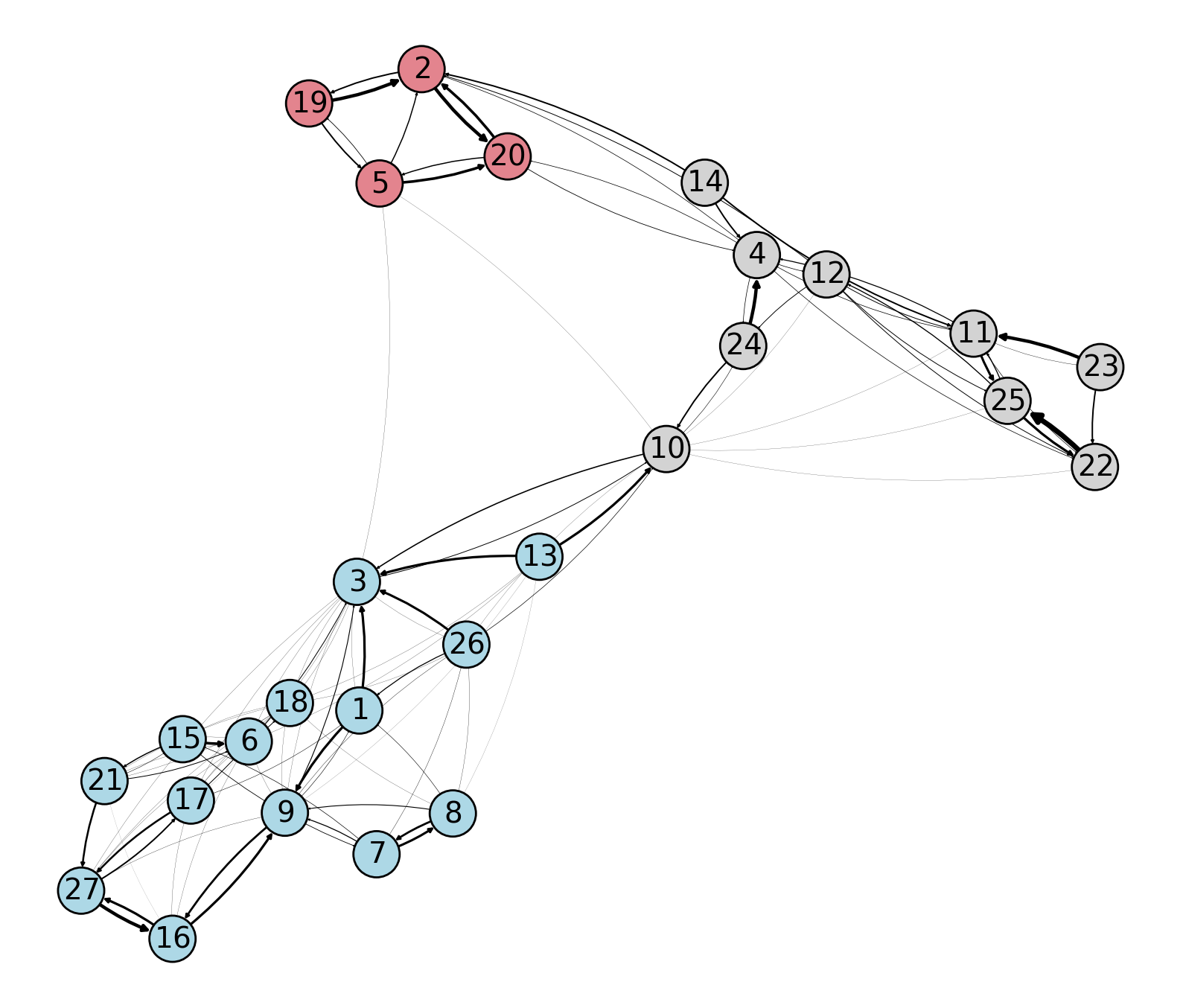}
    \caption{Social network of $N=27$ agents with modularity $\Phi_1(W_0) = 0.356 $ and Kemeny constant $\Phi_2(W_0) = 63.277$.
    Closely interconnected nodes are assigned to one out of three clusters. Thicker edges represent larger edge weights. }
    \label{fig:school_class}
\end{figure}

Features of interest in organizational networks, like the one in Figure~\ref{fig:school_class}, are related to connectivity or cohesiveness of the network, as they allow to assess the level of knowledge dissemination throughout an organization \citep{weenig1999communication, aral2007productivity,zhang2016information} and the effectiveness of the decision-making process \citep{cross2010organizational}.
For these reasons, we will in the following use features capturing modularity and connectedness.
For $ W \in {\cal W} $, let $ \Phi_1 ( W )  $ denote the modularity of $ W $ \citep{Newman_2006}, and let $ \Phi_2 (W) $ denote the connectivity of $ G $ expressed via the Kemeny constant \citep{kemeny1976finite} ($  \Phi_2 (W) = \infty $ in case $ W$ is not connected).
For both feature functions, we will provide a detailed definition later in Appendix~\ref{sf:mod}, \ref{sf:kemeny1}).

The question we now address is: what is the impact on the given network if the connectivity is increased to a given level?
This is a what-if analysis in which we adjust a single feature of a given network while aiming to minimally affect the overall network structure.
Suppose we want to adjust the connectivity feature for $W_0$ from the actual value $\Phi_2 ( W_0 ) $ to, say, $ \frac{1}{2} \Phi_2 ( W_0 ) $.
As the Kemeny constant is inversely related to network connectivity, this will result in a better-connected network.
To achieve this, we let $E = V^2 \setminus \{(i,i) : i \in V\}$ and apply $\mathcal{F}_{\left(\mathcal{W},p,\|K(W) - 31.638\|_2^2\right)}(w(0))$ for solving
\begin{align*}
    \begin{array}{rl}
    \min \quad &  \| w - w(0)\|_p \\
    \textup{s.t.} \quad & w \in {\cal W}_{K(W),31.638},
    \end{array}
\end{align*}
where $p \in \{1,2\}$, which provides insight into what (potentially non-existing) ties should be stimulated to improve connectivity.

First,  let us consider  the $L^2$-descent FBNC algorithm, where we obtain the (densely) constructed network
\begin{align*}
    \hat{W}_1 = \mathcal{F}_{\left(\mathcal{W}^M, 2, \|K(w) - 31.638\|_2^2\right)}\left(w(0)\right),
\end{align*}
as depicted in Figure~\ref{fig:schoolclass_fits} (left), where $\alpha = 10^{-3}$, $\gamma = 10^{-3}$ and $T=10000$, which is sufficiently large to reach convergence.
One could argue that less invasive network transformations in the number of link adjustments often provide much more insight into how the network structure must change to satisfy a certain feature (e.g., what links to stimulate to improve connectivity within an organization, instead of stimulating all non-present links).
Implicit regularization via the $L^1$-norm in this case has the benefit that the network is transformed using only a subset of links compared with the $L^2$-norm regularization, while still fitting the desired feature(s) in the network.
To that end, 
\begin{align*}
    \hat{W}_2 = \mathcal{F}_{\left(\mathcal{W}^M, 1, \|K(w) - 31.638\|_2^2\right)}\left(w(0)\right),
\end{align*}
see Figure~\ref{fig:schoolclass_fits} (right), for $\alpha = 10^{-3}$, $\gamma = 10^{-3}$ and $T=10000$.

\begin{figure}[t!]
\centering
\begin{minipage}{.475\textwidth}
  \centering
  \includegraphics[width=\linewidth]{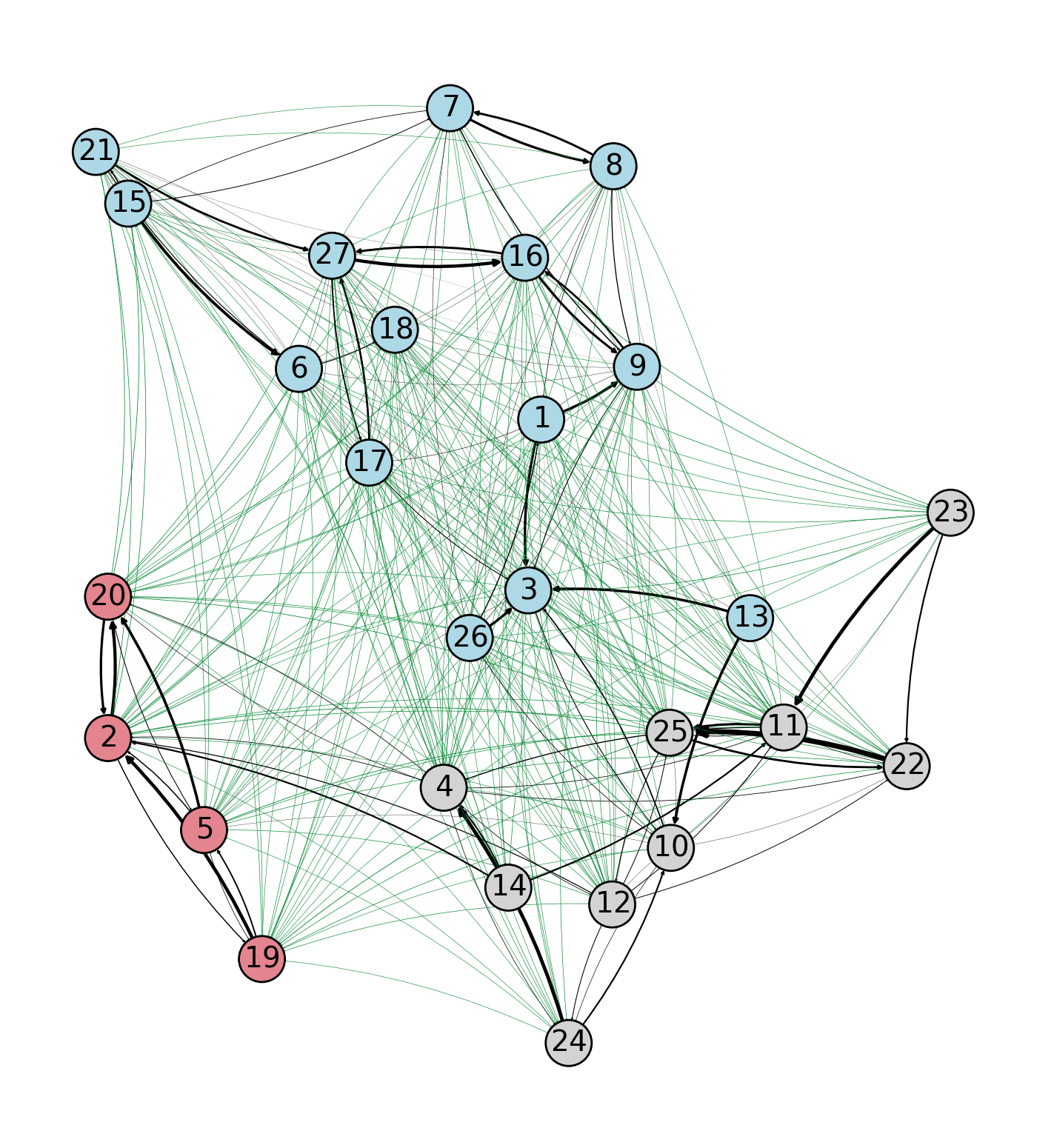}
\end{minipage}%
\hspace{0.5cm}
\begin{minipage}{.475\textwidth}
  \centering
  \includegraphics[width=\linewidth]{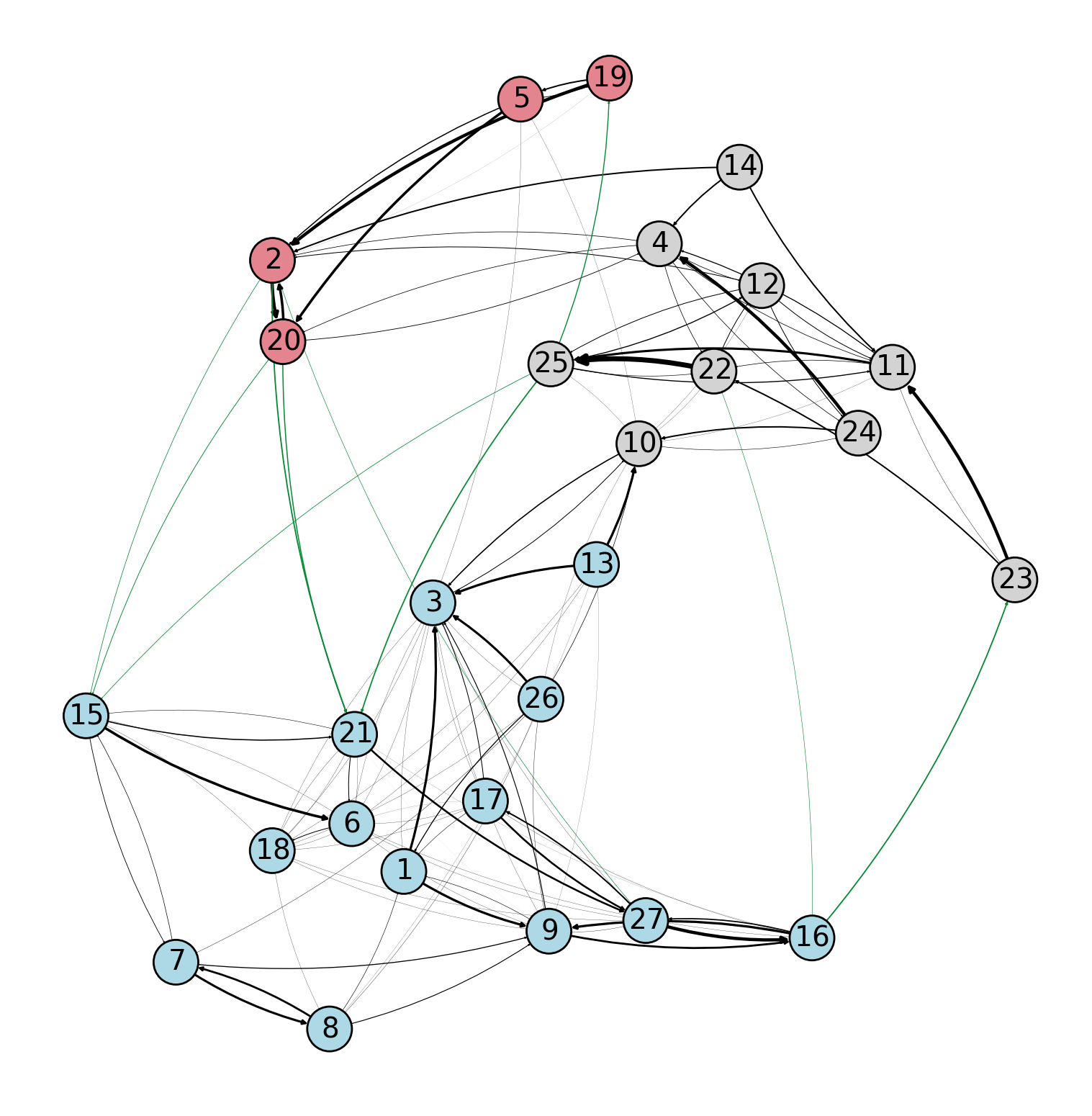}
\end{minipage}
\caption{Two adjustments $\hat{W}_1 $ (left) and $\hat{W}_2 $ (right) of $W_0$ (see Figure~\ref{fig:school_class}), with equivalent increase of connectivity, i.e., $\Phi_2 ( \hat{W}_1 ) = \Phi_2 ( \hat{W}_2 ) = 31.638 = \frac{1}{2}\Phi_2 ( W_0 ) $. 
New links are shown in green and made thicker for visibility. 
The densely constructed network $\hat{W}_1$ contains many new small weighted links, whereas the more parsimoniously constructed network $\hat{W}_2$ contains few (but more significant) new links.  }
\label{fig:schoolclass_fits}
\end{figure}

Applying the implicit regularization to the FBNC gradient search via the $L^1$-norm shows that it is possible to improve connectivity using only a small set of new links.
Note that our approach does not need a specification of costs for creating new links to obtain this parsimoniously constructed network.
In that sense, our model is parameter-free for improving graph connectivity where the (sparse) network transformation is a byproduct of choosing the $L^1$-norm regularization.
Moreover, the trajectory of our algorithm can serve as a road map for policy implementation to improve the connectivity in the network, as it subsequently adds links with the most potential to improve the connectivity. 
For the current instance, our algorithm created the following links in order $(1,26)$, $(15,22)$, $(15,21)$, $(1,20)$, $(24,20)$, $(19,20)$, $(24,14)$, $(24,18)$, $(1,14)$, $(19,14)$, so that $(1,26)$ can be considered to be the most fruitful potential link, followed by $(15,22) $, and so forth.

\subsection{A Financial Network Example}\label{subsec:financial_network_example}

Over the years, banks have become increasingly interconnected through the globalization of finance.
Understanding the effects of these interdependencies is of critical importance in managing so-called systemic risk, the risk of a strong economic downturn due to cascading failure in the financial sector, see \citet{hurd2016contagion,gandy2017bayesian} for various definitions.
To manage systemic risk, regulators need to identify features of a system that contribute to systemic risk.
For example, it has been shown that increasing exposure diversification can reduce systemic risk \citep{sachs2014completeness, Capponi_Chen_Yao_2016}.

Let us assume that a regulator wants to increase the exposure diversification on systemic risk indicators in an observed financial network $W_0$, where the nodes represent banks, and the links represent loans between banks.
In our analysis, we will not model banks' behavior nor specify a policy for how banks may be regulated to increase or decrease exposure diversification (e.g., the type of incentive for banks to increase or decrease exposure diversification), but instead, we leverage that our FBNC algorithm allows for producing versions of $W_0$ as if a policy stimulating exposure diversification is implemented.
Here, we assume that the policy is only successful if it is minimally disruptive on the market and therefore should preserve the network structure as much as possible.
The result of our approach is that we exclusively study the effects of regulating exposure diversification for a given financial network $W_0$, separately from the policy itself.

Let $W_0$ be the financial network used in Section~\ref{sec:num_example_sampling}, where now we postulate that the network is fully observable (including link weights), see Figure~\ref{fig:min_max_exposure} (left).
We define the economic needs (the interbank assets and liabilities) of each bank by the in-strengths and out-strengths. 
Furthermore, we let the set of links of the network be defined by the set of trading counterparties: $(i,j) \in E$ if $W_{i,j}>0$ and $(i,j) \notin E$ otherwise, and accordingly, only lending relationships are parameterized, see \eqref{eq:parameterization}.
Then, we measure the diversification of exposures using the concentration index $H(W)$, see \eqref{eq:HOSglobal}, where a high (low) concentration index implies low (high) exposure diversification.
For the network $W_0$, the concentration index $H(W_0) = 2.84$.

We will increase the diversification of exposures measured by the concentration index introduced in \eqref{eq:HOSglobal} of $W_0$,
while ensuring the following conditions in the resulting graph $\Hat{W}$:
\begin{enumerate}[label=\theenumi]
    \item \textbf{No interference in counterparties.} We ensure that banks are always exposed to the same or fewer counterparties, as ideally, banks should enjoy autonomy in choosing their counterparties without any regulatory interference.
    Therefore, let $W_0$ be parameterized by $w(0)$, using that $E = \{ (i,j) \in V^2 : (W_0)_{i,j} > 0\}$, so that $\Hat{w}_{i,j} = 0$ if $w(0)_{i,j} = 0$ for all $(i,j) \in V^2$.
    \item \textbf{Equivalent economic needs.} We match the total ingoing and outgoing exposures for each bank as any regulatory policy ideally does not affect the economic needs of each bank, i.e., for all $i \in V$ it holds that
    \[
             \sum_{j \in V_i} \hat{w}_{i,j} = \sum_{j \in V_i} w(0)_{i,j} \quad \mbox{and} \quad \sum_{j \in V_i} \Hat{w}_{j,i} = \sum_{j \in V_i} w(0)_{j,i}.
    \]
    \item \textbf{Minimal reconfiguration in the network.} We ensure that the resulting exposures are close to the initial financial network so that $ \| \Hat{w} - w(0) \|_2$ is small, as the implementability of financial regulations and policies strongly depends on the disruptive effects on the market.
    In other words, the resulting financial network structure should not deviate too much from the initial financial network, which means that we do not significantly alter the overall weights of the network.
\end{enumerate}

The use of the concentration index in \eqref{eq:HOSglobal} for measuring exposure diversification differs from previous studies investigating exposure diversification in financial networks who mainly focusing on the number of counterparties as the measure for diversification, e.g., see \citet{allen2000financial, nier2007network,elliott2014financial, ma2019diversification}.
However, having many counterparties does not imply a diversified set of exposures, as almost all lending can be concentrated at a single counterparty.
Therefore, we argue for the concentration index as it effectively captures the dispersion of loans; see \eqref{eq:HOSglobal}.
A similar measure for quantifying the diversification of loans to ours is used by \citet{sachs2014completeness} and \citet{Capponi_Chen_Yao_2016}.
\citet{sachs2014completeness} use an ensemble of networks with exposure sizes that have varying levels of entropy and subsequently test for systemic risk.
\citet{Capponi_Chen_Yao_2016}  
utilize the work of \cite{eisenberg2001systemic} to model contagion effects for networks with different levels of loans diversification measured by matrix majorization. 
Both works conclude that better diversification of loans can benefit systemic risk.

We let our FNBC algorithm increase exposure diversification in $W_0$ through imposing a lower concentration index; $H^* = \frac{3}{4} H(W_0) = 2.13$ (as to increase exposure diversification), by constructing the graph
\[
\hat{W} = \mathcal{F}_{\left(\left[ 0,\min \left( \max_i s^+_i, \max_i s^-_i \right)\right]^{|E|},2, \|H(w) - 2.13\|^2_2\right)}\left(w(0)\right).
\]
As by construction, we do not optimize the non-existing links, and so we satisfy (C1).
We will use the steepest feasible descent algorithm presented in \eqref{eq:descentalgo}, where we determine the steepest feasible descent direction $\delta^{\alpha}(w,2)$ while imposing the in-strengths and out-strengths as features through projection, thereby satisfying (C2).
Furthermore, choosing a small step size $\alpha = 0.05$ (and $\sigma = 0.5$) will implicitly minimize $ \| \Hat{w} - w_0 \|_2$, thus satisfying (C3).
Finally, we let $ \gamma = 10^{-3} $ in the termination condition and choose $T=50000$, which is large enough for convergence.
In Figure~\ref{fig:min_max_exposure}, we plot the transformation of $W_0$ to $\Hat{W}$.
The illustration highlights a noticeable pattern where larger loans decrease in magnitude, while smaller ones show a corresponding increase.

\begin{figure}
\centering
\begin{subfigure}[c]{.45\textwidth}
  \centering
  \includegraphics[width=\linewidth]{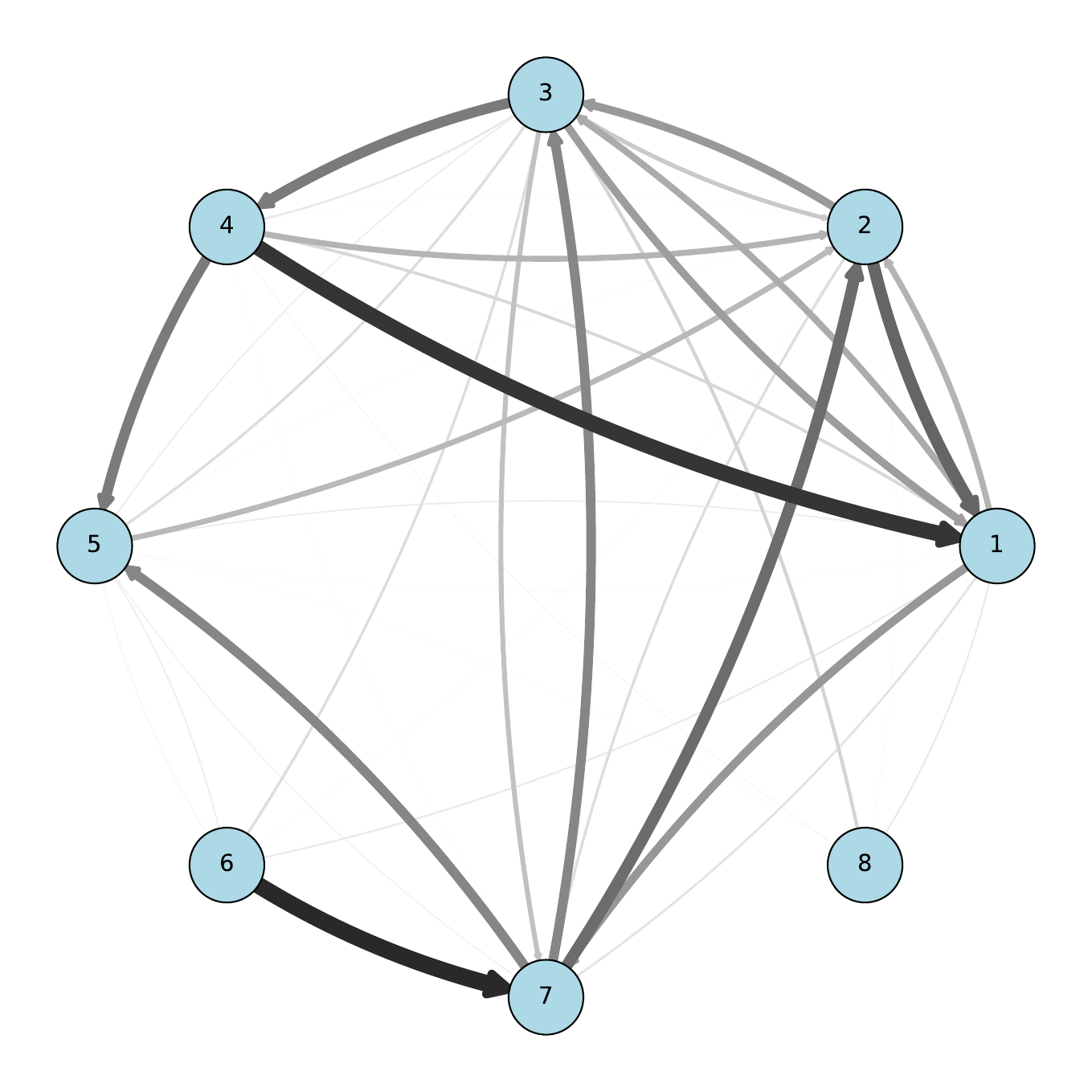}
\end{subfigure}%
\hfill
  \begin{subfigure}[c]{0.075\textwidth}
    \centering
    \includegraphics[width=0.75\linewidth]{Figures/arrow.jpg}
  \end{subfigure}
  \hfill
\begin{subfigure}[c]{.45\textwidth}
  \centering
  \includegraphics[width=\linewidth]{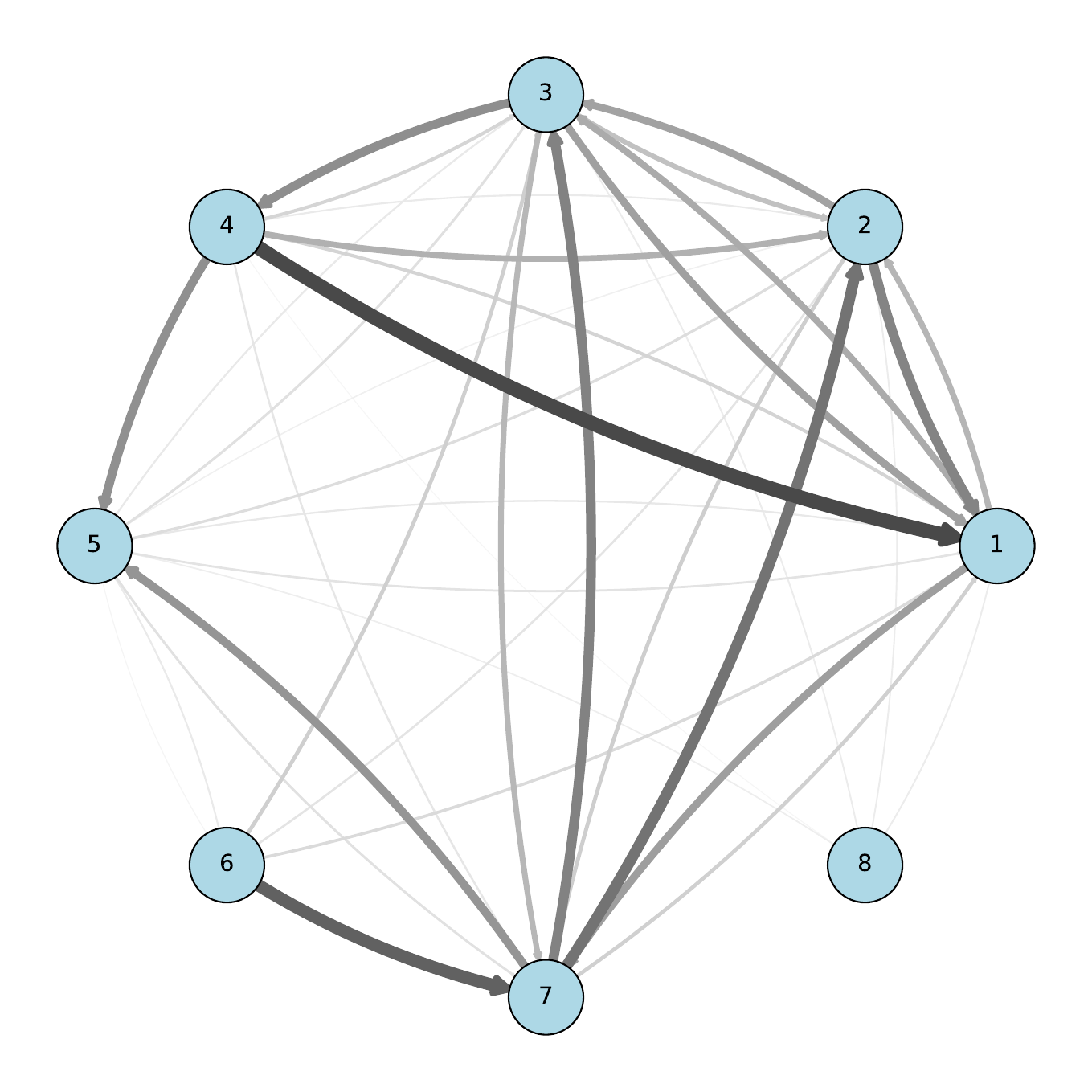}
\end{subfigure}%
\caption{$W_0$ (left, $H(W_0) = 2.84$) and $\Hat{W}$ (right, $H(\Hat{W}) = 2.13$) with increased exposure diversification while satisfying economic needs and counterparty constraints. Thicker and darker arrows correspond to larger link weights.}
\label{fig:min_max_exposure}
\end{figure}

The matrix $\Hat{W} - W_0$ specifically shows how the regulator should adjust $W_0$ to increase exposure diversification without interference in counterparties, preserving economic needs, and minimally reconfiguring the network.
Indeed, $\Hat{W} - W_0$ specifies the set of loans between banks to effectively reduce the exposure concentration between these banks.

Note that one may want a minimum for the concentration index $H(w)$, which one can find by using $J(w) = H(w)$. 
Here, we find ourselves in the context of classic optimization, where we can leverage the convexity of $H$, as shown in Lemma~\ref{lem:convexity}, allowing us to find a global minimizer for the concentration index for fixed out-strengths for each bank.
\begin{lemma} \label{lem:convexity}
    The concentration index as defined in \eqref{eq:HOSglobal} is a convex function in $W$, provided that $s^+_i(W) = s_i^{+*}$, for $i \in V$.
\end{lemma}
\begin{proof}
    Let us rewrite \eqref{eq:HOSglobal}
    \begin{align}
        H(W) = \sum_{(i,j) \in E} \left(\frac{W_{i,j}}{s^+_i(W)}\right)^2 = \sum_{(i,j) \in E} \frac{1}{\left(s_i^{+*}\right)^2} \sum_{j \in V_i} W_{i,j}^2.
    \end{align}
    By the assumption $s^+_i(W) = s^{+*}$, we obtain that $H(W)$ is a sum of convex functions, which shows that $H(W)$ is convex.
    This concludes the proof.
\end{proof}

\section{Conclusion} \label{sec:discussion}

In this paper, we introduced the Feature-Based Network Construction (FBNC) framework, a novel approach for constructing weighted digraphs that satisfy predefined structural features exactly. By leveraging gradient descent methods, FBNC provides a solution to the inverse problem of network (re)construction, enabling the generation of networks tailored to specific high-level features. This framework supports a range of applications, including random network generation, reconstruction of networks with anonymized representations, and imputation of incomplete network data. Additionally, FBNC's implicit regularization offers the ``what-if analysis'', a means for implementing targeted network features with minimal deviation from the original network.

Algorithmically, our approach is parameter-free, meaning that we can regularize the resulting graph in both the $L^2$- or $L^1$-norm without the need for introducing explicit regularization parameters.
In fact, we showed that explicit regularization methods like Ridge or Lasso regularization are infeasible approaches for realizing a desired feature in the network.
Furthermore, we show how the steepest feasible descent can be computed efficiently for bounded graphs and Markovian graphs.

Future work in hard constraint sampling of networks is to aim for larger network instances by increasing the scalability of the algorithm used, for example, via a steepest descent estimation with only two function evaluations \citet{Spall_1992}.
Additionally, the reconstruction of networks can benefit from a careful estimation of the link set $E$ through methods discussed by \citet{Squartini_Caldarelli_Cimini_Gabrielli_Garlaschelli_2018} in case the link set is not observable. 
Finally, it would be of interest to delve into the possibilities for obtaining a desired distribution of graphs, possibly by restraining the set of features and tailoring the distribution of initial networks.

\bibliography{References.bib}

\appendix

\addcontentsline{toc}{section}{Appendix}
\section{Structural Features} \label{sec:sf}

\subsection{Reciprocity} \label{sf:reciprocity}

Reciprocity is the tendency of a node $i$ to reciprocate an incoming link from node $j\neq i$, by forming a link towards node $j$. It is a network measure often used in social network analysis \citep{wasserman1994social}. Also, it is exclusively used in directed networks, as for undirected networks all links are reciprocated by construction. The commonly used measure for reciprocity is to count the number of reciprocated links and divide this by the total number of links \citep{wasserman1994social}. As this yields a discontinuous function, we cannot use this in our setting because of Condition \textbf{(A1)}. Therefore, we define a new continuous measure for reciprocity:
\begin{align*}
    r(W) = C^{-1} \sum_{(i,j) \in E : (j,i) \in E:i\leq j} \min\left(W_{i,j}, W_{j,i}\right),
\end{align*}

\noindent where we take the sum of the minimal values for each reciprocal relation and use some normalizing constant $C$.
As the min($\cdot, \cdot$) operator is non-differentiable, we can approximate the minimum in the following fashion:
\begin{align*}
    r(W) \approx \Tilde{r}(W) = C^{-1} \sum_{(i,j) \in E : (j,i) \in E:i\leq j} -\frac{1}{\xi} \log \left( \frac{ e^{-\xi W_{i,j}} + e^{-\xi W_{j,i}}}{N}\right)
\end{align*}

\noindent This approximation $\Tilde{r}(W)$ becomes asymptotically accurate to $r(W)$ for $\xi \xrightarrow{} \infty$. 
The minimum of this measure is 0, when there are no reciprocal relations and so $W_{i,j} = 0$ for all $i\neq j$. 
The maximum is $1$ when $W_{i,j} = W_{i,j}$ for all $i\neq j$.

\subsection{Triangular closures} \label{sf:3cycles}

\noindent Triangular closures measured using the number of 3-cycles \citep{luce1949method}. 
A 3-cycle exists if node $i$ connects with node $j\neq i$, given that node $j$ has a connection with node $k\neq j$ and node $k$ a connection with node $i\neq k$. 
It is an extension of reciprocity, which can be considered a 2-cycle, and is also often used in social network analysis \citep{holland1971transitivity, watts1998collective}:
\begin{align*}
    c(W) = C^{-1} \sum_{(i,j,k) : (i,j) \in E, (j,k) \in E, (k,i) \in E } \text{min}(W_{i,j}, W_{j,k}, W_{k,i}),
\end{align*}

\noindent where $C$ is chosen such that $c(W) \in [0,1]$.
We can approximate the min($\cdot, \cdot$) operator in a similar fashion as done for reciprocity in Section~\ref{sf:reciprocity}.

\subsection{Modularity} \label{sf:mod}

\noindent Modularity is a measure of a network's clustering \citep{Leicht_Newman_2007}, and is computed by 
\begin{align*}
    Q(W) = \frac{1}{2 m} \sum_{(i,j) \in E}\left[W(i, j) - \frac{s^-_i(W) s^+_j(W)}{2 m}\right] \delta_{c_i, c_j}, 
\end{align*}
where $m = \sum_{(i,j) \in E} W_{i,j}$, and $\delta_{c_i, c_j} = 1$ if both $i$ and $j$ share the same cluster and $0$ otherwise.

\subsection{Assortativity} \label{sf:assortativity}

\noindent Assortativity is a measure to describe the tendency for nodes to form connections with other nodes that are similar, also referred to as homophilic behavior \citep{newman2002assortative}. Generally, the degree of a node is used to express assortativity; however, other sorts of node attributes can be used too. Depending on the nature of the variable used to define similarity (categorical, ordinal, or continuous), we can construct different assortativity measures. Here, we will provide an example using the in-strength as (continuous) similarity measure, see (\ref{eq:s-}) and
(\ref{eq:s+}):
\begin{align*}
    \rho(W) = \frac{\sigma_{sr}(W)}{\sqrt{\sigma_{s}(W)\sigma_{r}(W)}},
\end{align*}
\vspace{-10pt}
\begin{align*}
    \sigma_{sr}(W) & := C^{-1} \sum_{(i,j) \in E} W_{i,j} (s_i^-(W) - \Hat{s}^+(W)) (s_j^-(W) - \Hat{s}^-(W)), \\[10pt]
    \sigma_{s}(W) &:= C^{-1} \sum_{(i,j) \in E} W_{i,j} (s_i^-(W) - \Hat{s}^+(W))^{2}, \\[10pt]
    \sigma_{r}(W) &:= C^{-1} \sum_{(i,j) \in E} W_{i,j} (s_j^-(W) - \Hat{s}^-(W))^{2},
\end{align*}

\noindent where $\sigma_{sr}$ denotes the weighted covariance between the in-strengths for senders and receivers, and $\sigma_{s}$ and $\sigma_{r}$ denote the weighted variances of the in-strengths for senders and receivers, respectively. 
Furthermore, $C = \sum_{(i,j) \in E} W_{i,j}$, but eventually cancels out, such that $\rho(W) \in [-1,1]$.
It follows that $\rho(W)$ is the weighted correlation between the in-strengths of the nodes.
Finally, $\Hat{s}^+$ is the weighted average of the in-strengths for all senders, and $\Hat{s}^-(W)$ is the weighted average of the in-strengths for all receivers:
\begin{align*}
    \Hat{s}^+(W) &= C^{-1} \sum_{(i,j) \in E} W_{i,j} s_{i}^-(W) = C^{-1} \sum_{i \in V} s_{i}^-(W) s_{i}^+(W), \\
    \Hat{s}^-(W) &= C^{-1} \sum_{(i,j) \in E} W_{i,j} s_{j}^-(W) = C^{-1} \sum_{i \in V} s_{i}^-(W)^2.
\end{align*}

\noindent If $\rho(W) = -1$, it implies nodes with a high in-strength tend to form connections to nodes with a low in-strength and vice versa. If $\rho(W) = 1$, the reverse is the case.

\subsection{Stationary distribution} \label{sf:stat}

\noindent Let $W \in \mathcal{W}^M$ such that it satisfies the conditions of stochasticity; that is, its entries are nonnegative, and its rows sum up to 1. 
Assuming $W$ is irreducible, we find its stationary distribution by solving the system of linear equations 
\begin{align} \label{eq:stat}
\pi = \pi W \quad \text{ and } \quad \sum_{i=1}^N \pi_i = 1.
\end{align}
Specifically, $\pi_i$ describes the long-run fraction of time a random walker modeled by $ W$ spends at a certain node $i$. Hence, it is commonly used as a measure of the importance of nodes in a network, such as used in Google's PageRank algorithm \citep{page1999pagerank}.

\subsection{Kemeny constant} \label{sf:kemeny1}

\noindent Assume that $W\in \mathcal{W}^M$ is stochastic and irreducible. 
The {\em Kemeny constant} is the expected time to reach a $\pi$-randomly drawn node $j$ from (any) node $i$ for transition matrix $W$.
It follows that a low (high) Kemeny constant implies a high (low) connected network.
It can be computed by:
\begin{align} \label{eq:kemeny2}
    K = \mathrm{tr}(D) + 1,
\end{align}
where $D$ denotes the deviation matrix of $W$. More formally,
\[
D = (I - W + \Pi)^{-1} - \Pi,
\]
and ergodic projecter $\Pi  = \lim_{m \to \infty} \frac{1}{m}\sum_{n=1}^m W^n$ is a matrix with rows corresponding to $\pi$, see \citet{kemeny1976finite}.
We refer the reader to \citet{berkhout2019analysis} for a more detailed discussion on the Kemeny constant as a connectivity measure.

\subsection{Effective graph resistance} \label{sf:effective_resistance}

\noindent The effective graph resistance $R$ can be expressed via a graph's Laplacian $L$, where 
\begin{align*}
    L_{i,j} = \begin{cases}
        -W_{i,j} & \text{ if } i \neq j; \\
        \sum_{j=1}^N W_{i,j} & \text{ if } i = j.
    \end{cases}
\end{align*}
Then, assuming the graph's weights are symmetric, $R$ can be written in terms of the non-negative eigenvalues $\mu_1 \leq \ldots \leq \mu_N$ of $L$ \citep{ellens2011effective}:
\begin{align*}
    R = N \sum_{i=2}^N \frac{1}{\mu_i}.
\end{align*}
Moreover, $R$ can be calculated by the mean first passage times of the random walker induced by the symmetric graph $W$ \citep{ghosh2008minimizing, ellens2011effective, franssen2024firstorder}. To that end, let the Markov chain $P(W)$ be defined via
\begin{align*}
    P_{i,j}(W) = \frac{W_{i,j}}{\sum_{j=1}^N W_{i,j}}.
\end{align*}
Then, the mean first passage time matrix of $P(W)$ is defined as
\begin{align*}
M = (I - D + {\bar 1} {\bar 1}^\top \cdot \textup{dg}(D)) \cdot \textup{dg}(\Pi )^{-1},
\end{align*} 
where $\bar{1}$ is a vector of ones of size $N$, and $\textup{dg}(\cdot)$ refers to the diagonal matrix constructed by positioning the elements of the input matrix along its diagonal. Then, $R$ can be written as the sum of mean first passage times:
\begin{align*}
    R = \frac{1}{\sum_{i=1}^N \sum_{j=1}^N W_{i,j}}\sum_{i \neq j} M_{i,j}.
\end{align*}

\end{document}